\documentclass[11pt,letterpaper]{article}

\usepackage[margin=1in]{geometry}
\usepackage[numbers,square]{natbib}

\usepackage{comment}

\usepackage{mathtools}
\usepackage{amsthm}
\usepackage{bm}

\usepackage[dvipsnames,usenames,table]{xcolor}
\usepackage[colorlinks=true,urlcolor=blue,linkcolor=RoyalBlue,citecolor=OliveGreen]{hyperref}

\usepackage{graphicx}
\usepackage{subcaption}

\usepackage[boxed]{algorithm2e}
\usepackage{amsmath,amssymb}
\usepackage{breqn}
\usepackage[nottoc]{tocbibind}

\newtheorem{theorem}{Theorem}
\newtheorem{lemma}[theorem]{Lemma}

\title{Online Primal Dual Meets Online Matching with Stochastic Rewards: Configuration LP to the Rescue}

\author{
    Zhiyi Huang\thanks{The University of Hong Kong. Email: zhiyi@cs.hku.hk}
    \and
    Qiankun Zhang\thanks{The University of Hong Kong. Email: qkzhang@connect.hku.hk}
}

\date{April, 2019}

\newcommand{\R}{\mathbf{R}}

\newcommand{\E}{\mathbf{E}}
\renewcommand{\Pr}{\mathbf{Pr}}

\renewcommand{\vec}[1]{\bm{#1}}

\newcommand{\alg}{\mathrm{ALG}}
\newcommand{\opt}{\mathrm{OPT}}

\SetKwInOut{Initialize}{initialize}

\begin{document}

\begin{titlepage}
\thispagestyle{empty}

\maketitle 

\begin{abstract}
    \thispagestyle{empty}
    Mehta and Panigrahi (FOCS 2012) introduce the problem of online matching with stochastic rewards, where edges are associated with success probabilities and a match succeeds with the probability of the corresponding edge.
It is one of the few online matching problems that have defied the randomized online primal dual framework by Devanur, Jain, and Kleinberg (SODA 2013) thus far.
This paper unlocks the power of randomized online primal dual in online matching with stochastic rewards by employing the configuration linear program rather than the standard matching linear program used in previous works.
Our main result is a $0.572$ competitive algorithm for the case of vanishing and unequal probabilities, improving the best previous bound of $0.534$ by Mehta, Waggoner, and Zadimoghaddam (SODA 2015) and, in fact, is even better than the best previous bound of $0.567$ by Mehta and Panigrahi (FOCS 2012) for the more restricted case of vanishing and equal probabilities.
For vanishing and equal probabilities, we get a better competitive ratio of $0.576$.
Our results further generalize to the vertex-weighted case due to the intrinsic robustness of the randomized online primal dual analysis.

\end{abstract}

\end{titlepage}

\section{Introduction}
\label{sec:introduction}

Online advertising platforms generate tens of billions of dollars of revenue per year and are the main driving application behind the research of online matching.
On such a platform, advertisers can be viewed as vertices on one side of a bipartite graph, while impressions (e.g., search queries to a search engine) can be viewed as vertices on the other side.
The former are known to the platform upfront and therefore are referred to as offline vertices, while the latter arrive online and therefore are referred to as online vertices.
On the arrival of each online vertex, the platform irrevocably matches it to an offline neighbor subject to certain constraints that vary depending on the model.

For example, the online bipartite matching model of \citet{KarpVV/STOC/1990} assumes that each online vertex can be matched to any offline vertex with whom it has an edge, which encodes, e.g., whether the impression is associated with a keyword that the advertiser is interested in, while each offline vertex can be matched at most once.
The goal is to maximize the cardinality of the matching.
The performance of an online algorithm is compared against the maximum cardinality matching in hindsight as in the standard competitive analysis, and the worst ratio with respect to (w.r.t.) an arbitrary bipartite graph is referred to as the competitive ratio of the algorithm.
\citet{KarpVV/STOC/1990} solve this model by proposing an algorithm called Ranking, which fixes a random ranking over the offline vertices at the beginning and then matches each online vertex to the unmatched neighbor with the highest rank.
They show that it is $1-\frac{1}{e} \approx 0.632$-competitive and this is the best possible.

\paragraph{Online Matching with Stochastic Rewards.}
Consider an online advertising platform on a search engine, however, the attempt to match an online vertex (i.e., a search query) and an offline vertex (i.e., an advertiser) may not succeed since the user may not click the ad.
The standard assumption in the literature is that the user will click the ad with a certain probability and this is called the click-through-rate.
If the attempt fails, the platform may try again to match the offline vertex later (the budget of the advertiser is still available), but the online vertex is gone.

To capture the above aspect of the problem, \citet{MehtaP/FOCS/2012} propose a model called online matching with stochastic rewards, where each edge is associated with a success probability.
On the arrival of an online vertex, the algorithm irrevocably matches it to an offline neighbor that has not been successfully matched yet.
Then, it succeeds with the probability associated with the edge, independent of whether previous edges are successful or not.
The objective is to maximize the number of offline vertices that are successfully matched.

\citet{MehtaP/FOCS/2012} consider two algorithms in the special case of equal probabilities, i.e., the success probabilities of edges are all equal to some $0 < p \le 1$.
The first algorithm is called Stochastic Balance, which greedily matches each online vertex to the neighbor that has the fewest match attempts in the past.
They use a factor revealing linear program (LP) to show that this algorithm is $\frac{1+e^{-2}}{2} \approx 0.567$-competitive as $p$ tends to zero, that is, the case of equal and vanishing probabilities.
The second one is a natural generalization of Ranking.
They show that it is at least $1 - \frac{2}{e} + \frac{2}{e^2} \approx 0.534$-competitive regardless of the value of $p$.%
\footnote{It in fact gets better as $p$ increases and tends to $1-\frac{1}{e} \approx 0.632$ as $p$ tends to $1$.}
\citet{MehtaP/FOCS/2012} also show an upper bounds (hardness results) of $0.588$ for Stochastic Balance, and $0.621$ for arbitrary algorithms.
In particular, the latter proves that this model is strictly harder than the online bipartite matching model of \citet{KarpVV/STOC/1990}, for which Ranking is $1-\frac{1}{e} \approx 0.632$-competitive.

Subsequently, \citet{MehtaWZ/SODA/2015} study the case of unequal and vanishing probabilities. 
That is, the success probabilities of the edges can be different but are upper bounded by some sufficiently small $0 < p \le 1$.
They propose an algorithm called Semi-Adaptive, and show that its competitive ratio tends to $0.534$ as $p$ tends to zero.

Closing the gaps between the upper and lower bounds, both for equal probabilities and for unequal probabilities, has been an interesting open problem.

\paragraph{Randomized Online Primal Dual.}
Shortly after the introduction of online bipartite matching with stochastic rewards, \citet{DevanurJK/SODA/2013} develop a unifying framework called randomized online primal dual that achieves the optimal $1-\frac{1}{e} \approx 0.632$ competitive ratio in the online bipartite matching model of \citet{KarpVV/STOC/1990} and a number of variants.
Informally speaking, it is a charging argument that splits the gain of each chosen edge between its two endpoints carefully, such that for any two neighboring vertices, the share of the gain they get in expectation is at least $1 - \frac{1}{e}$.
See Section~\ref{sec:online-primal-dual} for a formal discussion.
Given the success in the online bipartite matching model, it is natural to try applying randomized online primal dual to online matching with stochastic rewards to close the gaps between the upper and lower bounds. 
However, no competitive ratio other than the trivial $0.5$ is known using randomized online primal dual.%
\footnote{We know at least one expert in online matching who has an unsuccessful attempt in using online primal dual on online matching with stochastic rewards. The first author of this paper also tried and failed back in 2012.}
Therefore, other than being an interesting problem on its own, online matching with stochastic rewards also stands out as an important problem for understanding the power and limits of randomized online primal dual.

\subsection{Our Contributions}
\label{sec:contributions}

This paper shows how to unlock the power of randomized online primal dual in online matching with stochastic rewards, where the key is to use the configuration LP rather than the standard matching LP that was used in essentially all previous related works.
The main difference between these two LP's lies in their decision variables. 
The standard matching LP has a decision variable $x_{uv}$ for each edge $(u, v)$ that represents whether the edge is chosen by the algorithm.
In contrast, the configuration LP has for each offline vertex $u$ and each subset $S$ of its online neighbors a decision variable $x_{uS}$ that represents whether $S$ is the set of online vertices matched to $u$.

This difference in turn leads to different dual constraints, i.e., conditions that the subsequent charging arguments need to satisfy.
The charging argument using the standard matching LP has to split the gains such that for every edge $(u, v)$, the expected gain of the offline vertex $u$, scaled by the success probability $p_{uv}$ of the edge, plus the expected gain of the online vertex $v$ is at least the competitive ratio $\Gamma$ times $p_{uv}$.
In contrast, the argument using the configuration LP only needs to ensure that for each offline vertex $u$ and each subset $S$ of its online neighbors, the expected gains of $u$ and the online vertices in $S$ sum to at least $\Gamma$ times the sum of the success probabilities of the edges between $u$ and $S$.%
\footnote{For technical reasons that will become clear in the subsequent sections, we can in fact upper bound this sum of success probabilities by $1$ without loss of generality.}
The latter is a strictly weaker condition and allows for an implicit amortization among the online vertices in $S$.
See Section~\ref{sec:std-lp} for a formal discussion.

We focus on the case of vanishing probabilities and apply randomized online primal dual on the configuration LP of online matching with stochastic rewards.
We obtain better competitive ratios for the case of vanishing probabilities, for both equal and unequal probabilities.
The case of large success probabilities is left for future research.
Last but not least, our results generalize to the vertex-weighted case due to the intrinsic robustness of the randomized online primal dual analysis, while those by \citet{MehtaP/FOCS/2012} and \citet{MehtaWZ/SODA/2015} do not.
Table~\ref{tab:summary} presents a summary.

\begin{table}
    \centering
    \renewcommand{\arraystretch}{1.2}
    \begin{tabular}{|c|c|c|}
        \hline
        & Vanishing Equal Prob. & Vanishing Unequal Prob. \\
        \hline
        Best Previous Result (Unweighted) & $0.567$~\cite{MehtaP/FOCS/2012} & $0.534$~\cite{MehtaWZ/SODA/2015} \\
        \hline
        \textbf{This Paper (Vertex-weighted)} & $\bm{0.576}$ (Sec.~\ref{sec:equal-probabilities}) & $\bm{0.572}$ (Sec.~\ref{sec:unequal-probabilities}) \\
        \hline
        Hardness (Unweighted) & \multicolumn{2}{|c|}{$0.621$ ($0.588$ for Stochastic Balance)~\cite{MehtaP/FOCS/2012}} \\
        \hline
    \end{tabular}
    \caption{Summary of the results in this paper in comparison with previous ones}
    \label{tab:summary}
\end{table}

\paragraph{An Alternative Viewpoint.}
In order to present the technical contributions of this paper, we first introduce an alternative viewpoint of the problem proposed by \citet{MehtaP/FOCS/2012}, when the success probabilities tend to zero.
Instead of having for each edge an independent random variable that determines whether the edge succeeds or not, we can instead let each offline vertex $u$ independently sample a threshold $\theta_u$ from the exponential distribution with mean $1$.%
\footnote{More precisely, the thresholds are sampled from distributions that converge to the exponential distribution with mean $1$ at the limit when $p$ tends to zero. See Section~\ref{sec:reductions} for a formal discussion.}
Then, each offline vertex accumulates a load that equals the sum of the success probabilities of the edges matched to it, and succeeds at the moment the load exceeds the threshold.
Further, the expected number of successfully matched offline vertices equals the expected total loads of offline vertices.

\paragraph{Equal Probabilities.}
First consider vanishing and equal probabilities. 
We show that Stochastic Balance is at least $0.576$-competitive, improving the previous result of $0.567$ by \citet{MehtaP/FOCS/2012}. 
Further, we generalize the algorithm and its analysis to the vertex-weighted case.

\begin{theorem}
    \label{thm:equal-probabilities}
    There is a $0.576$-competitive online algorithm for (vertex-weighted) online matching with stochastic rewards with equal and vanishing probabilities.
\end{theorem}

In hindsight, we compare the structural lemmas in this paper with those used by \citet{MehtaP/FOCS/2012}.
The main difference that leads to the better competitive ratio lies in the improvement in an alternating path argument, which we briefly explain below.
Fix any offline vertex $u$ and the randomness related to offline vertices other than $u$, i.e., their thresholds.
First, consider the case when edges incident on $u$ never succeed, i.e., the threshold $\theta_u$ is sufficiently large, and suppose $\ell$ online vertices are matched to $u$ by Stochastic Balance in this case.
\citet{MehtaP/FOCS/2012} show that if the threshold $\theta_u$ was smaller and, thus, $k$-th attempt succeeded and the remaining $\ell - k$ vertices that were matched to $u$ must now be matched elsewhere or left unmatched, each of these vertices would trigger an alternating path and at most $\ell - k$ vertices would become unmatched as a result.
In contrast, our improved argument characterizes not only whether the online vertices are matched or not, but their matching qualities as well.
Here, the matching quality of an online vertex is measured by how many unsuccessful attempts has been made on the offline vertex that it matches, the smaller the better.
Concretely in the aforementioned setup, we show the followings.

\begin{itemize}
    \item \textbf{(Informal) Structural Lemma for Equal Probabilities:~} Fix any quality threshold, at most $\ell - k$ vertices change from better to worse than the threshold.
\end{itemize}

We believe this demonstrates the advantage of the randomized online primal dual framework, as it is unclear how to utilize this improved argument in the factor revealing LP approach.

\paragraph{Unequal Probabilities.}
For the more general vanishing and unequal probabilities, we introduce a randomized algorithm whose competitive ratio is at least $0.572$. 
It improves the best previous ratio of $0.534$ by \citet{MehtaWZ/SODA/2015} and, surprisingly, is even better than the best previous result of $0.567$ by \citet{MehtaP/FOCS/2012} for the more restricted case of vanishing and equal probabilities.

\begin{theorem}
    \label{thm:unequal-probabilities}
    There is a $0.572$-competitive randomized online algorithm for the problem of (vertex-weighted) online matching with stochastic rewards with vanishing and unequal probabilities. 
\end{theorem}

To understand the difficulties in solving the case of unequal probabilities, we briefly explain how the aforementioned alternating path argument breaks down completely.
Fix any offline vertex $u$ and the randomness related to offline vertices other than $u$, i.e., their thresholds. 
First, consider the set of online vertices that are matched to $u$ if it never succeeds;
let $v$ and $\tilde{v}$ be the last and second last online vertices matched to $u$.
Then, consider what would happen if $\tilde{v}$ succeeded and as a result $v$ could not be matched to $u$.
Unlike the case of equal probabilities, the changes this triggers are no longer an alternating path.
First, $v$ may switch to another offline vertex, say $u'$, with whom its success probability is twice as large as that with $u$, i.e., $p_{u'v} = 2 p_{uv}$.%
\footnote{Readers may wonder why the algorithm picks $u$ over $u'$ in the first place if the latter is more likely to succeed. This may be because $u'$ has a lot more unsuccessful attempts than $u$ has, which is an important factor in the decision makings of online algorithms including Stochastic Balance.}
If $u'$ succeeds before the change, however, having $v$ matched to $u'$ increases the load of $u'$.
As a result, the load of $u'$ may exceed the threshold earlier; 
the vertices that correspond to the last $p_{u'v} = 2 p_{uv}$ amount of load of $u'$ before the change, e.g., the last two vertices whose success probabilities with $u'$ are equal to $p_{uv}$, must be matched elsewhere.
Note that both the number of affected vertices and the sum of their success probabilities double.
As these vertices further trigger subsequent changes, the cascading effect cannot be bounded as in the aforementioned alternating path argument.

The Semi-Adaptive algorithm by \citet{MehtaWZ/SODA/2015} has two heavily interleaved components, an adaptive one and a non-adaptive one, to control the cascading effect. 
This approach, however, does not go well with the randomized online primal dual framework.

We take a completely different approach, which consists of two components.
The first one is a sequence of reductions showing that it suffices to design a deterministic algorithm for an easier fractional problem. 
In the fractional problem, the algorithm can fractionally split an online vertex among multiple offline vertices, so long as the total mass sum to at most $1$.
The increases in the loads of offline vertices are scaled accordingly.
Then, as in the aforementioned alternative viewpoint, an offline vertex succeeds at the moment that its load exceeds its threshold.
The idea is that given any deterministic algorithm for the fractional problem, referred to as the fractional algorithm, we can design a randomized algorithm for the original problem, referred to as the integral algorithm, that runs a copy of the fractional algorithm in the background, and simulates its decisions using independent rounding.
The key observation is that since the probabilities are tiny, it follows from standard concentration inequalities that with high probability the actual load of any offline vertex $u$ in the integral algorithm is close to its ``virtual load'' in the fractional algorithm in the background at all time.
We couple the thresholds in the two algorithms such that the threshold of any offline vertex $u$ in the fractional algorithm is slightly larger than that in the integral algorithm; 
so the simulation succeeds with high probability.
See Section~\ref{sec:reductions} for details of the reductions.

The second component is a competitive deterministic algorithm for the fractional problem.
The algorithm is greedy in natural that continuously assigns infinitesimal fraction of an online vertex $v$ to the offline neighbor with the highest ``score'', where the ``score'' of each offline neighbor $u$ is equal to $p_{uv} \big( 1 - f(\ell_u) \big)$ where $\ell_u$ is the current load of the offline vertex $u$, and $f(\cdot)$ is a function derived from the randomized online primal dual analysis.
If $f(\ell_u) = 1 - e^{-\ell_u}$, it becomes a fractional version of the algorithm suggested by \citet{MehtaP/FOCS/2012}.
Our analysis indicates this may not be the best option, however, and we optimize this function by solving a differential equation.

The key ingredient of our analysis is a set of invariants that are more robust than the aforementioned alternating path arguments used in previous works and in the case of equal probabilities.
Concretely, fix any offline vertex $u$, and the randomness related to the offline vertices other than $u$, i.e., their thresholds.
Then, as $u$'s threshold decreases, say, from $\theta_u$ to $\tilde{\theta}_u$, we have the following invariants about the offline and online sides respectively.

\begin{itemize}
    \item \textbf{Offline:~} The load of any offline vertex $u' \ne u$ at any given moment weakly increases.
    \item \textbf{Online:~} The dual variable of any online vertex, i.e., its share of the gain, weakly decreases.
\end{itemize}

Note that the offline invariant implies that the share that online vertices get from offline vertices other than $u$ actually weakly increases.
The only loss comes from the decrease of $u$'s threshold and, thus, its load.
Putting together with the online invariant, we get the structural lemma we need.

\begin{itemize}
    \item \textbf{(Informal) Structural Lemma for Unequal Probabilities:~} For any subset $S$ of online vertices, the expectation of the sum of their dual variables, i.e., their total share of the gain, decreases by at most the amount that online vertices, including those not in $S$, get from $u$ between load $\tilde{\theta}_u$ and $\theta_u$ before the decrease of $u$'s threshold.
\end{itemize}

This structural lemma is weaker than the one used in the case of equal probabilities.
Hence, the competitive ratio is also slightly smaller.
See Section~\ref{sec:unequal-probabilities-fractional-algorithm} for details of the fractional algorithm.

\subsection{Related Works}
\label{sec:related-works}

There is a vast literature on online matching.
Readers are referred to \citet{Mehta/FTTCS/2013} for a survey on this topic.
Below we briefly discuss those that are most relevant to this paper.

The analysis of Ranking is simplified by \citet{BirnbaumM/SIGACTNews/2008} and \citet{GoelM/SODA/2008}.
\citet{KalyanasundaramP/TCS/2000} show that a greedy algorithm, often referred to as Water-filling, achieves the optimal $1-\frac{1}{e}$ competitive ratio in online fractional matching.
\citet{MehtaSVV/FOCS/2005} study a variant called the AdWords problem and present a $1-\frac{1}{e}$ algorithm.
\citet{BuchbinderJN/ESA/2007} simplify its analysis using an online primal dual analysis.
\citet{DevanurJ/STOC/2012} further generalize this problem and give optimal online algorithms for online matching with concave returns.
\citet{AggarwalGKM/SODA/2011} investigate vertex-weighted online bipartite matching and generalize Ranking to obtain the optimal $1-\frac{1}{e}$ competitive ratio.
Most of these works, with the exception of \citet{DevanurJ/STOC/2012}, are unified under the randomized online primal dual framework by \citet{DevanurJK/SODA/2013}.

Another series of works seek to break the $1 - \frac{1}{e}$ barrier by considering random arrivals.
That is, even though the graph is arbitrary, the arrival order of the online vertices is uniformly at random.
\citet{KarandeMT/STOC/2011} and \citet{MahdianY/STOC/2011} independently study Ranking in the random arrival model and show lower bounds $0.653$ and $0.696$ respectively of the competitive ratio.
Recently, \citet{HuangTWZ/ICALP/2018} generalize the randomized online primal dual framework to handle random arrivals and give a generalization of Ranking that is $0.653$-competitive in the vertex-weighted case.
\citet{DevanurH/EC/2009} give a $1 - \epsilon$ competitive algorithm for AdWords in the random arrival model.

Much related to the random arrival model is the stochastic arrival model, where each online vertex is drawn from some known distribution independently.
It is a special case of random arrivals in the sense that any algorithm and its competitive ratio in the latter carry over to the former.
However, better results are generally possible in this model.
Building on a series of efforts by \citet{FeldmanMMM/FOCS/2009}, \citet{BahmaniK/ESA/2010}, and \citet{ManshadiGS/MOR/2012}, \citet{JailletL/MOR/2013} give the best known competitive ratio of $1 - 2 e^{-2} \approx 0.729$ for online bipartite matching with stochastic arrivals, and a slightly weaker ratio of $0.725$ for the vertex-weighted version.
\citet{HaeuplerMZ/WINE/2011} present a $0.667$-competitive algorithm for the edge-weighted problem.

Recently, there is a line of research by \citet{AshlagiBDJSS/arXiv/2018} and \citet{HuangKTWZZ/STOC/2018, HuangPTTWZ/SODA/2019} that consider a generalization of online bipartite matching called fully online matching, where the graph is not necessarily bipartite and all vertices arrive online. 
In particular, \citet{HuangKTWZZ/STOC/2018, HuangPTTWZ/SODA/2019} extend the randomized online primal dual framework to the fully online model to obtain tight competitive ratios of Ranking and Water-filling.

In concurrent and independent works, \citet{brubach2019vertex} and \citet{goyal2019online} consider variants of the online matching problem with stochastic rewards, and compare with alternative benchmarks other than the LP benchmark used in the original paper by \citet{MehtaP/FOCS/2012} and this paper.
As a result, their results are not directly comparable with ours.

\section{Preliminaries}
\label{sec:prelimaries}

\subsection{Model}
\label{sec:model}

Let $G = (U \cup V, E)$ be a bipartite graph, where $U$ and $V$ denote the left-hand-side (LHS) and the right-hand-side (RHS) of the graph respectively, and $E$ denotes the set of edges.
Each edge $(u, v)$ is associated with a success probability $0 \le p_{uv} \le 1$.
The LHS is given upfront. 
Each offline vertex $u \in U$ is associated with a nonnegative weight $w_u$.
Vertices on the RHS arrive one by one online.
On the arrival of an online vertex $v \in V$, its incident edges are released, and the algorithm irrevocably either matches it with an unmatched neighbor $u \in U$, which succeeds with probability $p_{uv}$, or leaves it unmatched. 
The goal is to maximize the expected number of successfully matched offline vertices, multiplied by their weights.

The rest of the paper uses the following nomenclature to avoid ambiguity. 
When the algorithm picks an edge $(u, v) \in E$ on the arrival of the online vertex $v$, we say that $u$ and $v$ are \emph{matched}.
Then, the match succeeds with probability $p_{uv}$ independently, and if so we say that $u$ \emph{succeeds} or is \emph{successful};
all offline vertices are \emph{unsuccessful} until the moments they succeed.

\paragraph{Alternative Viewpoint.}
We recall the alternative viewpoint of the unweighted case introduced by \citet{MehtaP/FOCS/2012}, which consists of two parts.
First, they observe that the expected number of successfully matched offline vertices equals the expected total load of offline vertices, where the load of an offline vertex $u$ is the sum of the success probabilities of its incident edges that the algorithm picks.

Second, by twisting the order of the randomness, the problem can be reinterpreted as follows.
At the beginning, each offline vertex $u$ independently draws a threshold $\theta_u$ from the exponential distribution with mean $1$.%
\footnote{More rigorously, it follows a distribution that converges to the exponential distribution with mean $1$ when $p$, the upper bound of success probabilities, tends to zero. See Section~\ref{sec:reductions} for a formal treatment of this matter.}
Then, each offline vertex $u$ succeeds at the moment its load exceeds the threshold $\theta_u$.
That is, we may interpret $\theta_u$ as a stochastic budget of $u$, whose realization is unknown to the algorithm until when the load exceeds the budget (if it happens at all).


\paragraph{Benchmark.}
Fix any instance.
The randomized matching given by any online algorithm can be viewed as a feasible solution to the following standard matching LP relaxation of the problem by letting $x_{uv}$ be the probability that edge $(u, v)$ is picked.
\begin{align}
    \textbf{StdLP:} \qquad \textrm{maximize} \quad & 
    \textstyle \sum_{(u,v) \in E} w_u \cdot p_{uv} \cdot x_{uv} \notag \\
	\textrm{subject to} \quad & 
    \textstyle \sum_{v : (u,v) \in E} p_{uv} \cdot x_{uv} \le 1 && \forall u \in U \label{eqn:stdlp-budget} \\
	& 
    \textstyle \sum_{u : (u,v) \in E} x_{uv} \le 1 && \forall v \in V \label{eqn:stdlp-capacity} \\
    & x_{uv} \ge 0 && \forall (u, v) \in E \label{eqn:stdlp-trivial}
\end{align}

Eqn.~\eqref{eqn:stdlp-budget} states that the expected load of each offline vertex $u$ is upper bounded by $1$, because it is upper bounded by the expectation of the threshold $\theta_u$, which equals $1$.
Eqn.~\eqref{eqn:stdlp-capacity} is a standard capacity constraint: each online vertex $v$ can be matched to at most one offline neighbor.
Following the vertex-weighted model of \citet{MehtaP/FOCS/2012}, we will consider the optimal objective of the above LP relaxation as the benchmark.

\subsection{Configuration LP}

Let $N_u$ denote the set of neighbors of an offline vertex $u$.
Further, for any $S \subseteq N_u$, let $p_u(S) = \min \{ \sum_{v \in S} p_{uv}, 1 \}$.
Consider the following offline-side configuration LP: 
\begin{align*}
    \textbf{ConfigLP:} \qquad \textrm{maximize} \quad & 
    \textstyle \sum_{u \in U} \sum_{S \subseteq N_u} w_u \cdot p_u(S) \cdot x_{uS} \\
    \textrm{subject to} \quad & 
    \textstyle \sum_{S \subseteq N_u} x_{uS} \le 1 & & \forall u \in U \\
    & 
    \textstyle \sum_{u \in U} \sum_{S \subseteq N_u : v \in S} x_{uS} \le 1 & & \forall v \in V \\
    & x_{uS} \ge 0 & & \forall u \in U, \forall S \subseteq N_u
\end{align*}

From an offline optimization viewpoint, as $p$, the upper bound of success probabilities, tends to zero, the configuration LP is equivalent to the standard matching LP in the following sense. 
First, the matching LP can be reinterpreted as to maximize $\sum_{u \in U} w_u \cdot \min \big\{ \sum_{v : (u, v) \in E} p_{uv} x_{uv}, 1 \big\}$ subject to only Eqn.~\eqref{eqn:stdlp-capacity} and \eqref{eqn:stdlp-trivial}.
On one hand, any feasible assignment of the configuration LP corresponds to a feasible assignment of the matching LP with $x_{uv} = \sum_{S \in N_u : v \in S} x_{uS}$.
Comparing the objectives of the LPs, the matching LP objective is weakly larger because $\min \{ z, 1 \}$ is concave.

On the other hand, given any feasible assignment of the matching LP, we can sample an integral matching via independent rounding, i.e., match each online vertex $v$ to a neighbor that is independently sampled according to $x_{uv}$'s.
Then, it gives an integral assignment of the configuration LP: $x_{uS} = 1$ for the subset $S$ of neighbors matched to $u$ in the independent rounding, and $0$ otherwise.
Further, by standard concentration inequalities (and union bound), the load of each vertex $u$ is at most its expected load, which is at most $1$ due to Eqn.~\eqref{eqn:stdlp-budget}, plus an additive error that diminishes as $p$ tends to zero.
Hence, the optimal of the configuration LP is at least that of the matching LP less a term that diminishes as $p$ tends to zero.

The corresponding dual LP of the configuration LP is the following:
\begin{align*}
    \textbf{Dual:} \qquad \textrm{minimize} \quad & 
    \textstyle \sum_{u \in U} \alpha_u + \sum_{v \in V} \beta_v \\
	\textrm{subject to} \quad & 
    \textstyle \alpha_u + \sum_{v \in S} \beta_v \ge w_u \cdot p_u(S) & & \forall u \in U, \forall S \subseteq N_u \\
	& \alpha_u, \beta_v \ge 0 & & \forall u \in U, \forall v \in V
\end{align*}

\subsection{Randomized Online Primal Dual}
\label{sec:online-primal-dual}

Next, we explain the randomized online primal dual framework of \citet{DevanurJK/SODA/2013}, instantiated with the above configuration LP.
We set the primal variables according to the matching decisions of the algorithm.
Further, we also maintain a dual assignment.
On the arrival of an online vertex $v$, there is a new dual variable $\beta_v$ and a new set of dual constraints related to $v$; 
we assign a value to $\beta_v$ and adjust the values of $\alpha_u$'s according to the matching decisions of the algorithm.
The next lemma summarizes a set of sufficient conditions for proving a competitive ratio of $\Gamma$.
\begin{lemma}
\label{lem:randomized-primal-dual}
    An algorithm is $\Gamma$-competitive if:
    \begin{enumerate}
        \item \emph{Equal primal and dual objectives:}
        \[
            \textstyle
            \sum_{u \in U} \sum_{S \subseteq N_u} p_u(S) \cdot w_u \cdot x_{uS} = \sum_{u \in U} \alpha_u + \sum_{v \in V} \beta_v 
            ~;
        \]
        \item \emph{Approximate dual feasibility:} For any $u \in U$ and any $S \subseteq N_u$, 
        \[
            \textstyle
            \E \big[ \alpha_u + \sum_{v \in S} \beta_v \big] \geq \Gamma \cdot w_u \cdot p_u(S) 
            ~.
        \]
    \end{enumerate}
\end{lemma}

The first condition states that the primal and dual objectives are equal.
In fact, we will ensure this by having the same amount of increments in the primal and dual objectives at all time.
Hence, the dual assignment can be viewed as a scheme of splitting the gain $p_{uv} \cdot w_u$ of each edge $(u,v)$ chosen by the algorithm among the vertices. 

\begin{proof}
    Let $\tilde{\alpha}_u = \Gamma^{-1} \cdot \E \big[\alpha_u\big]$ for all $u \in U$, and $\tilde{\beta}_v = \Gamma^{-1} \cdot \E\big[\beta_v\big]$ for all $v \in V$. 
    By the second condition, $\tilde{\alpha}_u$'s and $\tilde{\beta}_v$'s are a feasible dual assignment.
    Further, the first condition implies that the corresponding dual objective equals $\Gamma^{-1}$ times the expected objective of the algorithm. 
    Finally, by weak duality of LP, the objective of any feasible dual assignment is weakly larger than the optimal primal.
    Hence, the expected objective of the algorithm is at least a $\Gamma$ fraction of the optimal. 
\end{proof}

\subsection{Maximal-Bernstein-style Inequality}

Our argument for the case of vanishing and unequal probabilities will use the following maximal-Bernstein-style theorem, which follows, e.g., as a special case of Theorem 3 from \citet{Shao/JTP/2000}.

\begin{lemma}
    \label{lem:maximal-bernstein}
    Let $X_1, X_2, \dots, X_N$ be independent random variables with zero means and $|X_i| \le a$ for all $i \in [n]$.
    Further, suppose $M \ge 0$ satisfies that $\sum_{i=1}^N \E \big[ X_i^2 \big] \le M$.
    Then, for any $t \ge 0$:
    \[
        \Pr \bigg[ \max_{1 \le n \le N} \sum_{i=1}^n X_i \ge t \bigg] \le 2 \exp \bigg( - \frac{t^2}{4at + M} \bigg)
        ~.
    \]
\end{lemma}

\section{Equal Probabilities}
\label{sec:equal-probabilities}

This section considers the special case of vanishing and equal probabilities.
That is, we assume $p_{uv} = p$ for any edge $(u, v)$, and analyze the competitive ratio at the limit when $p$ tends to zero.

Concretely, recall the alternative viewpoint.
By changing the order of randomness, we can let each offline vertex $u$ independently sample a threshold $\theta_u$ that falls into interval $\big[ (i-1) p, ip \big)$ with probability $p (1 - p)^{i-1}$ for all positive integers $i$.  
Then, let $\ell_u$ denote the load of an offline vertex $u$, i.e., the sum of success probabilities of edges matched to $u$;
$u$ succeeds when its load exceeds its threshold.
Note that this distribution converges to the exponential distribution with mean $1$.
This section directly considers this distribution at the limit for ease of presentation.
See Section~\ref{sec:reductions} for a more formal treatment of this matter in the more general case of unequal probabilities.

As a warm-up, especially for readers who are unfamiliar with the randomized online primal dual analysis, we first consider the unweighted case, i.e., $w_u = 1$ for all $u \in U$. 
Section~\ref{sec:equal-probabilities-unweighted} presents an improved analysis of the Stochastic Balance algorithm by \citet{MehtaP/FOCS/2012}, which we include below for completeness.
The improved competitive ratio matches that stated in Theorem~\ref{thm:equal-probabilities}.

\begin{algorithm}
    \label{alg:equal-probabilities}
    \caption{Stochastic Balance}
    \For{each online vertex $v \in V$}{
        match $v$ to the unsuccessful neighbor with the least load, breaking ties lexicographically
    }
\end{algorithm}

Then, Section~\ref{sec:equal-probabilities-weighted} proves Theorem~\ref{thm:equal-probabilities} by generalizing the algorithm and the online primal dual analysis to the vertex-weighted case.

\subsection{Unweighted Case}
\label{sec:equal-probabilities-unweighted}

In the unweighted setting, i.e., $w_u = 1$ for any $u \in U$, the configuration LP and its dual become:
\begin{align*}
    \textbf{ConfigLP:} \qquad \textrm{maximize} \quad & 
    \textstyle \sum_{u \in U} \sum_{S \subseteq N_u} p_u(S) \cdot x_{uS} \\
    \textrm{subject to} \quad & 
    \textstyle \sum_{S \subseteq N_u} x_{uS} \le 1 & & \forall u \in U \\
    & 
    \textstyle \sum_{u \in U} \sum_{S \subseteq N_u : v \in S} x_{uS} \le 1 & & \forall v \in V \\
    & x_{uS} \ge 0 & & \forall u \in U, \forall S \subseteq N_u \\[2ex]
    \textbf{Dual:} \qquad \textrm{minimize} \quad & 
    \textstyle \sum_{u \in U} \alpha_u + \sum_{v \in V} \beta_v \\
	\textrm{subject to} \quad & 
    \textstyle \alpha_u + \sum_{v \in S} \beta_v \ge p_u(S) & & \forall u \in U, \forall S \subseteq N_u \\
	& \alpha_u, \beta_v \ge 0 & & \forall u \in U, \forall v \in V
\end{align*}


\subsubsection{Dual Assignment}
\label{sec:equal-probabilities-dual-unweighed}

We now explain how to maintain a dual assignment.
Recall that $\ell_u$ denotes the load of an offline vertex $u \in U$. 
Note that $u$'s load increases over time; let it be the current load in the following discussion. 
For some non-decreasing function $f : \R^+ \mapsto \R^+$ to be determined in the analysis, we maintain the dual assignment as follows.

\begin{enumerate}
    \item Initially, let $\alpha_u = 0$ for all $u \in U$.
    \item On the arrival of an online vertex $v \in V$ and, thus, the corresponding dual variable $\beta_v$:
    \begin{enumerate}
        \item If $v$ is matched to $u \in U$, increase $\alpha_u$ by $p f(\ell_u)$ and let $\beta_v = p \big( 1 - f(\ell_u) \big)$. 
        \item If $v$ has no unsuccessful neighbor, let $\beta_v = 0$.
    \end{enumerate}
\end{enumerate}

Let $F(\ell) = \int_0^\ell f(z) dz$ denote the integral of $f$.
Note that we have $\alpha_u = F(\ell_u)$ in the limit as $p$ tends to zero.

\subsubsection{Randomized Primal Dual Analysis}
\label{sec:equal-probabilities-primal-dual-unweighted}

Note that the dual assignment, in particular, step 2, splits the gain of $p$ from each matched edge to the dual variables of the two endpoints. 
Hence, the first condition in Lemma~\ref{lem:randomized-primal-dual} holds by definition.
The rest of this subsection shows that the second condition in the lemma holds with the desired competitive ratio $\Gamma$, i.e., for any $u$ and any $S \subseteq N_u$, over the randomness of the thresholds:
\[
    \E \bigg[\alpha_u + \sum_{v\in S} \beta_v \bigg] \geq \Gamma \cdot p_u(S)
    ~.
\]
    

Similar to the randomized online primal dual analysis of online matching problems in previous works, we will fix the randomness related to offline vertices other than $u$, i.e., their thresholds, denoted as  $\vec{\theta}_{-u}$.
We will show the inequality over the randomness of $u$'s threshold $\theta_u$ alone.

Then, first consider what happens if $\theta_u$ is sufficiently large so that vertex $u$ remains unsuccessful throughout the algorithm; 
we abuse notation and denote this case by $\theta_u = \infty$.
Let $\ell_u^\infty$ denote the load of $u$ in this case.
Next, we characterize the matching related to $u$ and online vertices in $S$ for different realization of the threshold $\theta_u$ of $u$.
\begin{itemize}
    \item \textbf{Case 1: $\theta_u \ge \ell_u^\infty$.~}
        In this case, the matching is the same as the $\theta_u = \infty$ case since $u$ has enough budget to accommodate all the load matched to it.  
    \begin{itemize}
        \item The load of $u$ equals $\ell_u^\infty$.
        \item All online vertices in $S$ are matched to offline neighbors with load at most $\ell_u^\infty$ at the time of the matches, since $u$ is available with load at most $\ell_u^\infty$.
    \end{itemize}
    \item \textbf{Case 2: $0 \le \theta_u < \ell_u^\infty$.~}
        In this case, $u$ can only accommodate the load up to its threshold. 
        Some online vertices that were matched to $u$ when $\theta_u = \infty$ need to be matched elsewhere.
    \begin{itemize}
        \item The load of $u$ equals $\theta_u$.
        \item All online vertices in $S$, except for up to $p^{-1} \big( \ell_u^\infty - \theta_u \big)$ of them, are matched to offline neighbors with load at most $\ell_u^\infty$ at the time of the matches (Lemma~\ref{lem:equal-probabilities-structural}). 
    \end{itemize}
\end{itemize}

The above characterization is identical to that in \citet{MehtaP/FOCS/2012} on the offline side, but is better on the online side. 
\citet{MehtaP/FOCS/2012} only analyze whether the online vertices are matched or not, while we further inspect their matching qualities in terms of the load of the neighbors that they match to (at the time of the matches).
The first three bullets follow directly from the definition of the algorithm.
The fourth one is the main structural lemma whose proof is deferred to Section~\ref{sec:equal-probabilities-structural}.

\paragraph{Contribution from $\alpha_u$.}
%
We now proceed to prove approximate dual feasibility using the above characterization. 
First, consider the contribution from the offline vertex $u$.

\begin{lemma}
    \label{lem:alpha-equal-prob}
    For any thresholds $\vec{\theta}_{-u}$ of offline vertices other than $u$ and the corresponding $\ell_u^\infty$:
    \[
        \textstyle
        \E _{\theta_u} \big[ \alpha_u | \vec{\theta}_{-u} \big] \geq \int_0^{\ell_u^\infty} e^{-\theta_u} f(\theta_u) d\theta_u 
        ~.
    \]
\end{lemma}

\begin{proof} 
    By the above characterization, $u$'s load equals $\ell_u^\infty$ when $\theta_u \ge \ell_u^\infty$, and equals $\theta_u$ when $0 \le \theta_u < \ell_u^\infty$.
    Further, recall that $\alpha_u = F(\ell_u)$ and that $\theta_u$ follows the exponential distribution with mean $1$.
    We get that:
    \begin{align*}
        \E _{\theta_u} \big[ \alpha_u | \vec{\theta}_{-u} \big] 
        &
        \textstyle
        \geq \int_0^{\ell_u^{\infty}} e^{-\theta_u} F(\theta_u) d\theta_u + e^{-\ell_u^{\infty}} F(\ell_u^{\infty}) \\
        &
        \textstyle
        = \int_0^{\ell_u^\infty} e^{-\theta_u} f(\theta_u) d\theta_u 
        ~.
    \end{align*}

    Here, the equality follows by integration by parts.
\end{proof}

\paragraph{Contribution from $\beta_v$'s.}
Next, consider the contribution from the online vertices $v \in S$.
For ease of notations, we let $z^+$ denote $\max \{ z, 0 \}$.

\begin{lemma}
    \label{lem:beta-equal-prob}
    For any thresholds $\vec{\theta}_{-u}$ of offline vertices other than $u$ and the corresponding $\ell_u^\infty$:
    \[
        \textstyle
        \E_{\theta_u} \big[ \sum_{v \in S} \beta_v | \vec{\theta}_{-u} \big] 
        \geq \left( e^{-\ell_u^{\infty}} p_u(S) + \int_0^{\ell_u^{\infty}} e^{-\theta_u} \big( p_u(S) - (\ell_u^{\infty}-\theta_u) \big)^+ d\theta_u \right) \big( 1 - f(\ell_u^{\infty}) \big) 
        ~.
    \]
\end{lemma}

\begin{proof}
    By the above characterization, all vertices $v \in S$ are matched to neighbors with load at most $\ell_u^\infty$ at the time of the matches when $\theta_u \ge \ell_u^\infty$.
    This happens with probability $e^{-\ell_u^\infty}$ and contributes: 
    \[
        \textstyle
        e^{-\ell_u^{\infty}} p_u(S) \cdot \big( 1 - f(\ell_u^{\infty}) \big)
        ~,
    \]
    to the expectation in the lemma. 

    Similarly, all vertices $v \in S$, except for $p^{-1} \big( \ell_u^\infty - \theta_u \big)$ of them, are still matched to neighbors with load at most $\ell_u^\infty$ at the time of the matches when $0 \le \theta_u < \ell_u^\infty$.
    This contributes:
    \[
        \textstyle
        \int_0^{\ell_u^{\infty}} e^{-\theta_u} \big( p_u(S) - (\ell_u^{\infty}-\theta_u) \big)^+ d\theta_u \cdot \big( 1 - f(\ell_u^{\infty}) \big)
        ~,
    \]
    to the expectation in the lemma. 
    Putting together proves the lemma
\end{proof}

\begin{proof}[Proof of Theorem~\ref{thm:equal-probabilities} (Unweighted Case)]
    By Lemma~\ref{lem:alpha-equal-prob} and Lemma~\ref{lem:beta-equal-prob}, we have:
    \[
        \textstyle
        \E \big[ \alpha_u + \sum_{v \in S} \beta_v \big] \ge
        \int_0^{\ell_u^\infty} e^{-\theta_u} f(\theta_u) d\theta_u + \left( e^{-\ell_u^{\infty}} p_u(S) + \int_0^{\ell_u^{\infty}} e^{-\theta_u} \big( p_u(S) - (\ell_u^{\infty}-\theta_u) \big)^+ d\theta_u \right) \big( 1 - f(\ell_u^{\infty}) \big) 
        ~.
    \]
    
    To show approximate dual feasibility, it suffices to find a non-decreasing function $f$ such that for any $\ell_u^{\infty} \geq 0$ and any $0 \leq p_u(S) \leq 1$:
    \begin{equation}
        \label{eqn:de-equal}
        \textstyle
        \int_0^{\ell_u^\infty} e^{-\theta_u} f(\theta_u) d\theta_u + \left( e^{-\ell_u^{\infty}} p_u(S) + \int_0^{\ell_u^{\infty}} e^{-\theta_u} \big( p_u(S) - (\ell_u^{\infty}-\theta_u) \big)^+ d\theta_u \right) \big( 1 - f(\ell_u^{\infty}) \big)  
        \geq \Gamma \cdot p_u(S)
    \end{equation}
    
    Solving for the largest competitive $\Gamma$ for which this set of differential inequalities admits a feasible solution (see Section~\ref{sec:solve-de-equal} for details), we get that for $g(x) = \frac{1}{2-x-e^{-x}}$, and $h(x) = \exp \big( \int_x^1 g(z)dz \big)$:
    \begin{equation}
        \label{eqn:f-equal}
        f(x) = \left\{
        \begin{array}{ll}
            1-\frac{1}{e} & \quad x > 1 \\
            \frac{1}{h(x)} \left( 1- \frac{1}{e} + \int_x^1 (1-e^{-y}) g(y) h(y) dy \right) & \quad 0 \leq x \leq 1
        \end{array}
        \right.
    \end{equation}
    satisfies Eqn.~\eqref{eqn:de-equal} with the competitive ratio $\Gamma = 1-f(0) \approx 0.576$.
\end{proof}

\subsubsection{Structural Lemma for Equal Probabilities}
\label{sec:equal-probabilities-structural}

To formally describe the structural lemma, let us consider an alternative viewpoint of the matching process of Stochastic Balance. 
Given a graph $G = (U \cup V, E)$, and any thresholds $\vec{\theta}$, define a graph $G(\vec{\theta}) = (U' \cup V, E')$ as follows:
\begin{itemize}
    \item For each offline vertex $u \in U$ in the original graph, let there be $p^{-1} \theta_u$ copies of $u$ in $U'$, and denote them as $(u,0), (u,p), (u,2p), \dots, (u,\theta_u-p)$.
        That is, since $u$ can accommodate $p^{-1} \theta_u$ online vertices before it succeeds, let there be this many copies in the alternative viewpoint.
    \item For any online vertex $v$ that has an edge $(u, v)$ in the original graph, let there be edges between it and all copies of $u$ in $G(\vec{\theta})$.
\end{itemize}

Then, each copy of an offline vertex can be matched only once;
each online vertex is matched to a copy with the smallest load, breaking ties lexicographically.
Further, changing the threshold of a particular offline vertex $u$ is equivalent to adding/removing copies in the alternative viewpoint.

    
Moreover, fix any offline vertex $u \in U$.
Partition the online vertices into \emph{good} and \emph{bad} vertices according to their matching qualities.
Concretely, let $V_G(\vec{\theta})$ denote the set of \emph{good} online vertices that are matched to copies of offline vertices with load less than $\ell_u^\infty$;
let $V_B(\vec{\theta})$ denote the remaining \emph{bad} online vertices.
Then, the structual lemma can be stated as follows.
    
\begin{lemma}
    \label{lem:equal-probabilities-structural}
    Suppose $\theta_u < \ell_u^{\infty}$.
    Then, the matchings w.r.t.\ $G \big( \infty,\vec{\theta}_{-u} \big)$ and $G \big( \theta_u,\vec{\theta}_{-u} \big)$ satisfy that $V_G \big( \infty, \vec{\theta}_{-u} \big)$ and $V_G \big( \theta_u, \vec{\theta}_{-u} \big)$ differ by at most $p^{-1} \big( \ell_u^{\infty}- \theta_u \big)$ vertices.
\end{lemma}

In fact, we will prove a more general version of the lemma so that it can be used in the more general vertex-weighted case.

\paragraph{Order-based Algorithms.}
%
Consider any total order over the $(u, \ell)$ pairs, where $u \in U$ and $\ell$ is a nonnegative multiple of $p$, such that for any $u$, and any $\ell < \ell'$, $(u, \ell)$ ranks higher than $(u, \ell')$.
There is a corresponding \emph{order-based algorithm} that matches each online vertex to an unsuccessful offline neighbor $u$, which together with its current load, i.e., $(u, \ell_u)$, has the highest rank.
In particular, Stochastic Balance is the special case when the total ordering is deduced from the natural ordering of load $\ell_u$, breaking ties in lexicographical ordering of $u$.

\paragraph{Order-based Partition of Good and Bad Vertices.}
Consider any order-based algorithm. 
Fix any offline vertex $u \in U$ and the corresponding $\ell_u^{\infty}$, i.e., the load of $u$ when $\theta_u = \infty$.
Define the partition of \emph{good} and \emph{bad} online vertices by letting $V_G(\vec{\theta})$ be the set of good online vertices that match to an offline copy $(u', \ell_{u'})$ which ranks higher than $(u, \ell_u^{\infty})$ in the order of matching process, and $V_B(\vec{\theta})$ be the remaining bad online vertices.



\begin{figure}
    \centering
    \includegraphics[width=.85\textwidth]{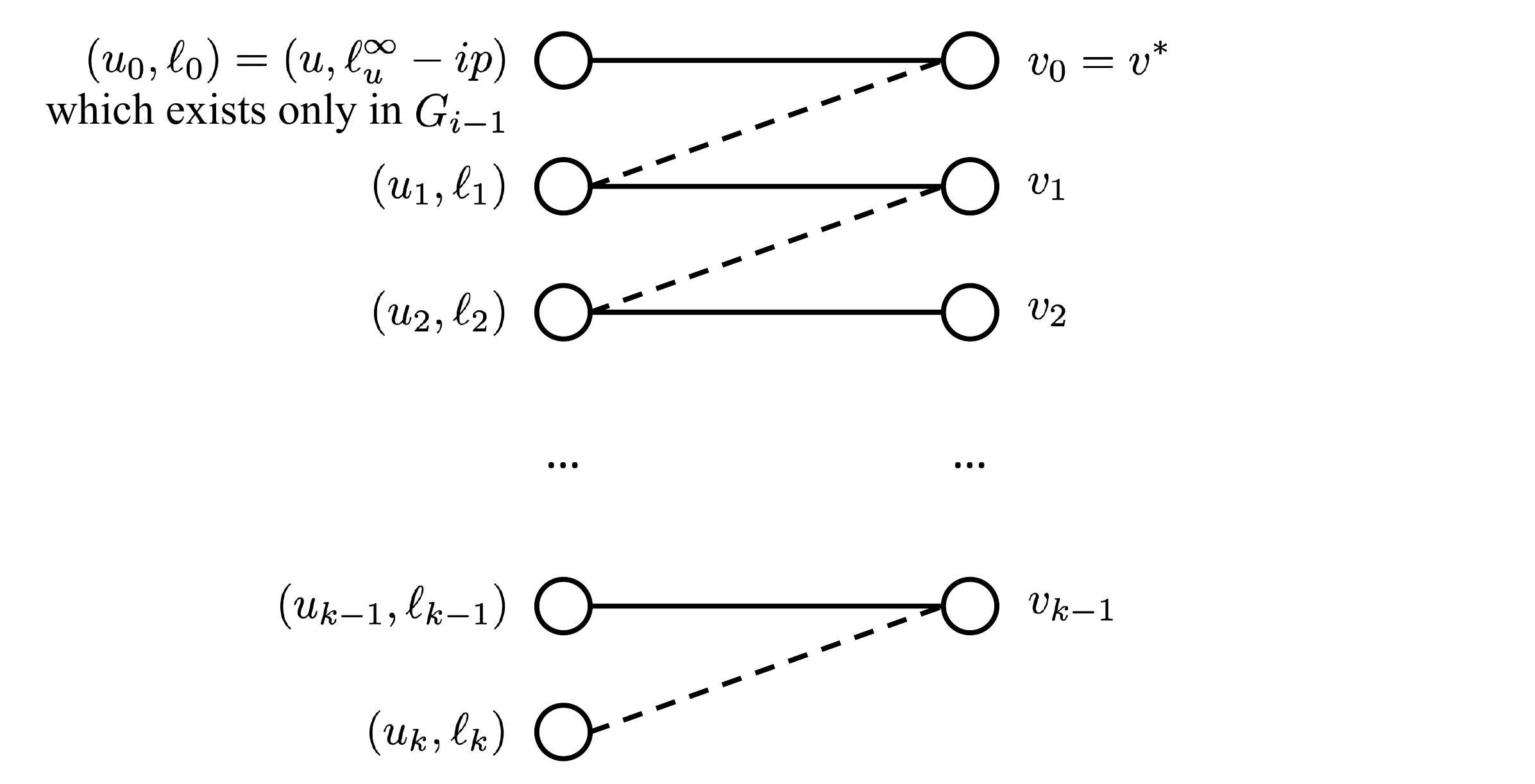}
    \caption{An illustrative picture of an alternating path from comparing $G_i$ and $G_{i-1}$. Solid edges are those in the matching in $G_{i-1}$, while dashed edges are those in the matching in $G_i$. In this example, the path ends with a copy of an offline vertex changes from unmatched to matched.}
    \label{fig:alternating-path}
\end{figure}

\begin{lemma}
    \label{lem:equal-probabilities-structural-general}
    Consider any order-based algorithm and any order-based partition of good and bad vertices that corresponds to $u$ and $\ell_u^\infty$. 
    Suppose $\theta_u < \ell_u^{\infty}$.
    Then, the matchings w.r.t.\ $G \big( \infty,\vec{\theta}_{-u} \big)$ and $G \big( \theta_u,\vec{\theta}_{-u} \big)$ satisfy that $V_G \big( \infty, \vec{\theta}_{-u} \big)$ and $V_G \big( \theta_u, \vec{\theta}_{-u} \big)$ differ by at most $p^{-1} \big( \ell_u^{\infty}- \theta_u \big)$ vertices.
\end{lemma}

\begin{proof}
    Note that the matchings of $G \big(\ell_u^\infty, \vec{\theta}_{-u} \big)$ and $G \big(\infty, \vec{\theta}_{-u} \big)$ are the same by the definition of $\ell_u^\infty$. 
    It suffices to show that for any positive integer $i$, graph $V_G \big( \ell_u^\infty - ip, \vec{\theta}_{-u} \big)$ and $V_G \big( \ell_u^\infty - (i-1)p , \vec{\theta}_{-u} \big)$ differ by at most one vertex. 
    Then, combining the cases of $1 \le i \le p^{-1} \big( \ell_u^\infty - \theta_u \big)$ proves the lemma.

    Next, let $G_i = G \big( \ell_u^\infty - ip, \vec{\theta}_{-u} \big)$ and $G_{i-1} = G \big( \ell_u^\infty - (i-1)p, \vec{\theta}_{-u} \big)$, and $V_i = V_G \big( \ell_u^\infty - ip, \vec{\theta}_{-u} \big)$ and $V_{i-1} = V_G \big( \ell_u^\infty - (i-1)p, \vec{\theta}_{-u} \big)$ for ease of notations.
    Note that $G_i$ and $G_{i-1}$ differ only in that $G_{i-1}$ has one more copy of vertex $u$, i.e., $(u, \ell_u^\infty - ip)$.
    Let $v^*$ be the vertex matched to this extra copy of $u$ in $G_{i-1}$.

    Before the arrival of $v^*$, the matchings on both graphs are identical.
    After the arrival of $v^*$, it triggers an alternating path.
    In $G_{i-1}$, $v^*$ is matched to $(u, \ell_u^\infty - ip)$;
    for notation consistency in the rest of the argument, let $u_0 = u$, $\ell_0 = \ell_u^\infty - ip$, and $v_0 = v^*$ and, thus, $v_0$ is matched to $(u_0, \ell_0)$.
    In $G_i$, however, $(u_0, \ell_0)$ does not exist and therefore $v_0 = v^*$ is matched to a copy of some other vertex, say, $(u_1, \ell_1)$.
    Subsequently, the vertex that is matched to $(u_1, \ell_1)$ in $G_{i-1}$, denoted as $v_1$ is matched to some other copy, say, $(u_2, \ell_2)$ in $G_i$.
    In general, for any $j$, the vertex $v_j$ that is matched to $(u_j, \ell_j)$ in $G_{i-1}$ is matched to some other copy $(u_{j+1}, \ell_{j+1})$ in $G_{i-1}$, except for at the end of the alternative path.
    There are two ways this alternating path ends.
    Either that for some integer $k$, the vertex $v_k$ that is matched to $(u_k, \ell_k)$ in $G_{i-1}$ is now unmatched in $G_i$, i.e., the alternating path ends with a vertex on the online side that changes from matched to unmatched, or that for some integer $k$, $(u_k, \ell_k)$ is unmatched in $G_{i-1}$, i.e., the alternating path ends with a copy of an offline vertex changes from unmatched to matched.
    See Figure~\ref{fig:alternating-path} for an illustrative picture.

    Further, note that the online vertices prefer $(u_j, \ell_j)$ over $(u_{j+1}, \ell_{j+1})$ for any $j$, in other words, $(u_j, \ell_j)$ ranks higher than $(u_{j+1}, \ell_{j+1})$ in the total order.
    Recall that $(u_k, \ell_k)$ is the last LHS vertex on the alternating path.
    We get that in the sequence of copies $(u_j, \ell_j)$ for $j \le k$, their ranks are non-increasing w.r.t.\ the order of matching process. 
    Thus, there is at most one index $j$ such that $(u_j, \ell_j)$ ranks higher than $(u, \ell_u^\infty)$ and $(u_{j+1}, \ell_{j+1})$ ranks lower or equal to $(u, \ell_u^\infty)$.
    In other words, there is at most one online vertex that changes from good to bad as a result of the alternating path.
    
    In the unweighted case, i.e., Lemma~\ref{lem:equal-probabilities-structural}, we get $\ell_0 \le \ell_1 \le \dots \le \ell_k$ in the sequence of copies.
    Therefore, there is at most one index $j$ such that $\ell_j < \ell_u^\infty$ and $\ell_{j+1} \ge \ell_u^\infty$. 
\end{proof}

\subsubsection{Solving the Differential Inequality}
\label{sec:solve-de-equal}

This section shows how to solve the differential inequality in Eqn.~\eqref{eqn:de-equal} to get the solution $f(x)$ in Eqn.~\eqref{eqn:f-equal} the corresponding competitive ratio.
For ease of notations, we write $\ell_u^\infty$ as $\ell$, $p_u(S)$ as $p$, and $\theta_u$ as $\theta$. 
Then, Eqn.~\eqref{eqn:de-equal} becomes:
%
\[
    \textstyle
    \int_0^{\ell} e^{-\theta} f(\theta) d\theta + \left( e^{-\ell} p + \int_0^{\ell} e^{-\theta} \big( p - (\ell-\theta) \big)^+ d\theta \right) \big( 1 - f(\ell) \big) 
    \geq \Gamma \cdot p
    ~.
\]

First, note that the $z^+$ operation effectively restricts the second integration in the above equation to be from $\max \big\{0, \ell - p\big\}$ to $\ell$. 
This allows us to simplify the above equation by taking away the $z^+$ operation.
Concretely, it suffices to ensure that for any $0 \le p < \ell$:
\begin{equation}
    \label{eqn:de-equal-p<l}
    \textstyle
    \int_0^{\ell} e^{-\theta} f(\theta) d\theta 
    + \left( 1-f(\ell) \right) (e^{-\ell+p}-e^{-\ell})
    \geq \Gamma \cdot p
    ~,
\end{equation}
and for any $0 \le \ell \le p$:
\begin{equation}
    \label{eqn:de-equal-p>l}
    \textstyle
    \int_0^{\ell} e^{-\theta} f(\theta) d\theta 
    + \left( 1-f(\ell) \right) (1-\ell+p-e^{-\ell})
    \geq \Gamma \cdot p
    ~.
\end{equation}

\paragraph{Left Boundary Condition.}
We start with the left boundary condition of $f$ at zero.

\begin{lemma}
    \label{lem:equal-probabilities-de-left-boundary}
    For the solution $f$ to the differential inequalities, we have:
    \[
        f(0) \le 1 - \Gamma
        ~.
    \]
\end{lemma}

\begin{proof}
    It follows by letting $\ell = 0$ and, e.g., $p = 1$, in Eqn.~\eqref{eqn:de-equal-p>l}.
\end{proof}

In the following discussion, we focus on the functions such that the left boundary condition holds with equality, i.e., $f(0) = 1 - \Gamma$.

\paragraph{Unit Mass of Online Neighbors.}
Then, we show that the worst-case scenarios are when there is a unit mass of online neighbors, i.e., when $p = 1$.

\begin{lemma}
    If $f$ is a non-decreasing function such that $f(0) = 1 - \Gamma$ and Eqn.~\eqref{eqn:de-equal-p<l} holds for any $\ell > p = 1$, while Eqn.~\eqref{eqn:de-equal-p>l} holds for any $\ell \le p = 1$.
    Then, $f$ satisfies the equations for any $0 \le p \le 1$.
\end{lemma}

\begin{proof}
    Consider both sides of Eqn.~\eqref{eqn:de-equal-p<l} and Eqn.~\eqref{eqn:de-equal-p>l} as functions of $p$. 
    It suffices to show that the derivative of the RHS is weakly larger than that of the LHS for any $0 \le p \le 1$.
    The derivative of the RHS is $\Gamma$, while that of the LHS is at most:
    \begin{align*}
        e^{-\ell+p} \big( 1 - f(\ell) \big) 
        \le & 1 - f(\ell) 
        && \text{($p \le l$)} \\[1ex]
        \le & 1 - f(0) 
        && \text{($f$ is non-decreasing)} \\[1ex]
        = & \Gamma
        && \text{(left boundary condition of $f$)}
        ~.
    \end{align*}
    for Eqn.~\eqref{eqn:de-equal-p<l}, and at most:
    \[
        1 - f(\ell)
        \leq 1 - f(0)
        \leq \Gamma
    \]
    for Eqn.~\eqref{eqn:de-equal-p>l}. 
    So the lemma follows.
\end{proof}
    
    Then it suffices to solve the following simplified differential inequalities subject to the boundary condition that $f(0) = 1 - \Gamma$.
    For any $\ell > 1$, we need:
    \begin{equation}
        \label{eqn:de-equal-l>1}
        \textstyle
        \int_0^{\ell} e^{-\theta} f(\theta) d\theta 
        + \left( 1-f(\ell) \right) (e^{-\ell+1}-e^{-\ell})
        \geq \Gamma 
        ~,
    \end{equation}
    while for any $0 \le \ell \le 1$, we need:
    \begin{equation}
        \label{eqn:de-equal-l<1}
        \textstyle
        \int_0^{\ell} e^{-\theta} f(\theta) d\theta 
        + \left( 1-f(\ell) \right) (2-\ell-e^{-\ell})
        \geq \Gamma 
    \end{equation}
    
    \paragraph{Right Boundary Condition.~}
To further simplification, we impose a right boundary condition that it suffices to ensure that $f(\ell) = 1 - \frac{1}{e}$ for $\ell \ge 1$.

\begin{lemma}
    \label{lem:equal-probabilities-de-right-boundary}
    Suppose $f : [0, 1] \mapsto \R^+$ is a non-decreasing function such that $f(0) = 1 - \Gamma$, $f(1) = 1 - \frac{1}{e}$, and Eqn.~\eqref{eqn:de-equal-l<1} holds for any $0 \le \ell \le 1$.
    Then, the extension of $f$ such that $f(\ell) = f(1) = 1 - \frac{1}{e}$ for all $\ell > 1$ satisfies Eqn.~\eqref{eqn:de-equal} for all $\ell \ge 0$.
\end{lemma}
\begin{proof}
    By the assumption of $f$, it suffices to show that the difference between the RHS of Eqn.~\eqref{eqn:de-equal-l>1} when with any $\ell > 1$ and that when $\ell = 1$ is non-negative.
    The difference is (recalling that $f(\ell) =1-1/e$ for all $\ell \geq 1$):
    \begin{align*}
    &  \textstyle
    ( \int_0^{\ell} e^{-\theta}f(\theta)d\theta + (1-f(\ell))(e^{-\ell+1}-e^{-\ell}) ) - ( \int_0^{1} e^{-\theta}f(\theta)d\theta + (1-f(1))(1-e^{-1}) ) \\
    & \textstyle \qquad
    = \int_1^\ell e^{-\theta} f(1) d\theta + \big( 1 - f(1) \big) \big( ( e^{-\ell+1} - e^{-\ell} ) - (1 - e^{-1}) \big) \\
    & \textstyle \qquad
    = \big(1 - e^{-\ell+1} \big) \big( f(1) - (1 - e^{-1}) \big) \\
    & \textstyle \qquad
    = 0
        ~.
    \end{align*}
\end{proof}

As a result, it suffices to find a non-decreasing function $f$ subject to the boundary conditions $f(0) =1-\Gamma$ and $f(1)=1-1/e$, such that Eqn.\eqref{eqn:de-equal-l<1} holds for any $0 \leq \ell \leq 1$.
Solving this differential inequality for $f$ gives,
\begin{align*}
    \textstyle
    f(x) =\frac{1}{h(x)} ( 1- \frac{1}{e} + \int_x^1 (1-e^{-y}) g(y) h(y) dy )
    ~,
\end{align*}
where $g(x) = \frac{1}{2-x-e^{-x}}$, and $h(x) = \exp \big( \int_x^1 g(z)dz \big)$. 

To conclude, we find the $f(x)$ shown in Eqn.\eqref{eqn:f-equal}. Thus $\Gamma = 1-f(0) \approx 0.576$.

\subsection{Generalization to the Vertex-weighted Case}
\label{sec:equal-probabilities-weighted}

This subsection generalizes Stochastic Balance to the vertex-weighted case, and completes the proof of Theorem~\ref{thm:equal-probabilities}.
The Algorithm~\ref{alg:equal-probabilities-weighted} below for vertex-weighted case is driven by the randomized online primal dual framework.

\begin{algorithm}
    \caption{A deterministic algorithm for the vertex-weighted case}
    \label{alg:equal-probabilities-weighted}
    \For{each online vertex $v \in V$}{
        match $v$ to the unsuccessful neighbor with the largest $w_u \cdot p \cdot (1-f(\ell_u))$, where $f$ is defined in Eqn.~\eqref{eqn:f-equal}, breaking ties lexicographically.
    }
\end{algorithm}

In particular, recall that gain sharing rule in the unweighted case in Section~\ref{sec:equal-probabilities-dual-unweighed}.
When the algorithm picks an edge $(u, v)$, the gain of $p$ from the edge is split among $u$ and $v$ according to the current load $\ell_u$ of the offline vertex $u$ and some non-decreasing function $f$:
the dual variable of $u$ increases by $p f(\ell_u)$;
the dual variable of $v$ is set to be $p \big( 1 - f(\ell_u) \big)$.
Stochastic Balance can be viewed as matching the online vertex $v$ to the offline neighbor that gives $v$ the largest share of the gain due to $f$ is non-decreasing.
Now consider the vertex-weighted case.
Following the same reason, suppose $u$ and $v$ are matched, the gain $w_u p$ on the edge is split between $u$ and $v$ such that the dual variable of $u$ increases $w_u p f(\ell_u)$ and the dual variable of $v$ is set to $w_u p \big( 1 - f(\ell_u) \big)$.
The algorithm matches each online vertex $v$ to the offline neighbor that offers the largest share of gain to $v$, i.e., the one with the largest value of $w_u p \big( 1 - f(\ell_u) \big)$.


Note that the Algorithm~\ref{alg:equal-probabilities-weighted} is an order-based algorithm as defined in Section~\ref{sec:equal-probabilities-structural}. 
Consider the viewpoint of matching process and $G(\vec{\theta})$ introduced in Section~\ref{sec:equal-probabilities-structural}, the rank of each copy $(u, \ell_u)$ is deduced from the value of $w_up(1-f(\ell_u))$, that is, Algorithm~\ref{alg:equal-probabilities-weighted} matches an offline vertex $v$ to a copy $(u, \ell_u)$ with the highest rank, i.e., with the greatest $w_up(1-f(\ell_u))$. 
Recall that we partition the online vertices to good and bad by whether the rank of the offline copy they match is higher than $(u, \ell_u^{\infty})$.
In the vertex-weighted case, good vertices are matched to an offline copy $(u', \ell_{u'})$, where $w_{u'}p(1-f(\ell_{u'})) > w_up(1-f(\ell_u^{\infty}))$, and bad vertices are the remaining ones.

\subsubsection{Dual Assignment}
\label{sec:equal-probabilities-weighted-dual}

The dual assignment procedure follows the idea from the unweighted case mentioned above, except for sharing a gain $w_u p$ instead of $p$ in unweighted case.
Detailed procedure is as follows:

\begin{enumerate}
    \item Initially, let $\alpha_u = 0$ for all $u \in U$.
    \item On the arrival of an online vertex $v \in V$ and, thus, the corresponding dual variable $\beta_v$:
    \begin{enumerate}
        \item If $v$ is matched to $u \in U$, increase $\alpha_u$ by $w_u p f(\ell_u)$ and let $\beta_v = w_u p \big( 1 - f(\ell_u) \big)$. 
        \item If $v$ has no unsuccessful neighbor, let $\beta_v = 0$.
    \end{enumerate}
\end{enumerate}

Let $F(\ell) = \int_0^\ell f(z) dz$ denote the integration of $f$.
Note that we have $\alpha_u = w_u F(\ell_u)$ in the limit as $p$ tends to zero.

\subsubsection{Randomized Primal Dual Analysis}

The dual assignment similarly splits the gain of $w_up$ from each matched edge to the dual variables of the two endpoints. 
Hence, the first condition in Lemma~\ref{lem:randomized-primal-dual} holds by definition.
The rest of this subsection shows that the second condition in the lemma holds with the desired competitive ratio $\Gamma$, i.e., for any $u$ and any $S \subseteq N_u$, over the randomness of the thresholds:
\[
    \E \bigg[\alpha_u + \sum_{v\in S} \beta_v \bigg] \geq \Gamma \cdot w_u \cdot p_u(S)
    ~.
\]
    


%
Similar to the unweighted case, we fix $\vec{\theta}_{-u}$, and characterize the matching related to $u$ and online vertices in $S$ for different realization of the threshold $\theta_u$ of $u$. Recall that $\ell_u^\infty$ denotes the load of $u$ when $u$ remains unsuccessful at the end of the algorithm.

\begin{itemize}
    \item \textbf{Case 1: $\theta_u \ge \ell_u^\infty$.~}
        In this case, the matching is the same as the $\theta_u = \infty$ case since $u$ has enough budget to accommodate all the load matched to it.  
    \begin{itemize}
        \item The load of $u$ equals $\ell_u^\infty$.
        \item All online vertices in $S$ are matched to offline copies with the value at least $w_u p (1-f(\ell_u^\infty))$ at the time of the matches from the viewpoint in Section~\ref{sec:equal-probabilities-structural}, since $u$ is available with load at most $\ell_u^\infty$, i.e., the copy $(u, \ell_u^\infty)$ is available.
    \end{itemize}
    \item \textbf{Case 2: $0 \le \theta_u < \ell_u^\infty$.~}
        In this case, $u$ can only accommodate the load up to its threshold. 
        Some online vertices that were matched to $u$ when $\theta_u = \infty$ need to be matched elsewhere.
    \begin{itemize}
        \item The load of $u$ equals $\theta_u$.
        \item All online vertices in $S$, except for up to $p^{-1} \big( \ell_u^\infty - \theta_u \big)$ of them, are matched to offline copies with the value at least $w_u p (1-f(\ell_u^\infty))$ at the time of the matches (Lemma~\ref{lem:equal-probabilities-structural-general}). 
    \end{itemize}
\end{itemize}

Then we get the gain of $\alpha_u$ and $\beta_v$ in a similar way as Section~\ref{sec:equal-probabilities-primal-dual-unweighted}.

\paragraph{Contribution from $\alpha_u$.}
%
We now proceed to prove approximate dual feasibility using the above characterization. 
First, consider the contribution from the offline vertex $u$.

\begin{lemma}
    \label{lem:alpha-equal-prob-weighted}
    For any thresholds $\vec{\theta}_{-u}$ of offline vertices other than $u$ and the corresponding $\ell_u^\infty$:
    \[
        \textstyle
        \E _{\theta_u} \big[ \alpha_u | \vec{\theta}_{-u} \big] \geq \int_0^{\ell_u^\infty} e^{-\theta_u} w_u f(\theta_u) d\theta_u 
        ~.
    \]
\end{lemma}

\begin{proof} 
    By the above characterization, $u$'s load equals $\ell_u^\infty$ when $\theta_u \ge \ell_u^\infty$, and equals $\theta_u$ when $0 \le \theta_u < \ell_u^\infty$.
    Further, recall that $\alpha_u = w_u F(\ell_u)$ and that $\theta_u$ follows the exponential distribution with mean $1$.
    We get that:
    \begin{align*}
        \E _{\theta_u} \big[ \alpha_u | \vec{\theta}_{-u} \big] 
        &
        \textstyle
        \geq \int_0^{\ell_u^{\infty}} e^{-\theta_u} w_u F(\theta_u) d\theta_u + e^{-\ell_u^{\infty}} w_u F(\ell_u^{\infty}) \\
        &
        \textstyle
        = \int_0^{\ell_u^\infty} e^{-\theta_u} w_u f(\theta_u) d\theta_u 
        ~.
    \end{align*}

    Here, the equality follows by integration by parts.
\end{proof}

\paragraph{Contribution from $\beta_v$'s.}
Next, consider the contribution from the online vertices $v \in S$.
For ease of notations, we let $z^+$ denote $\max \{ z, 0 \}$.

\begin{lemma}
    \label{lem:beta-equal-prob-weighted}
    For any thresholds $\vec{\theta}_{-u}$ of offline vertices other than $u$ and the corresponding $\ell_u^\infty$:
    \[
        \textstyle
        \E_{\theta_u} \big[ \sum_{v \in S} \beta_v | \vec{\theta}_{-u} \big] 
        \geq w_u \left( e^{-\ell_u^{\infty}}  p_u(S) + \int_0^{\ell_u^{\infty}} e^{-\theta_u}  \big( p_u(S) - (\ell_u^{\infty}-\theta_u) \big)^+ d\theta_u \right) \big( 1 - f(\ell_u^{\infty}) \big) 
        ~.
    \]
\end{lemma}

\begin{proof}
    By the above characterization, all vertices $v \in S$ are matched to an offline copy with a value at least $w_u p (1-f(\ell_u^\infty))$ at the time of the matches when $\theta_u \ge \ell_u^\infty$.
    This happens with probability $e^{-\ell_u^\infty}$ and contributes: 
    \[
        \textstyle
        e^{-\ell_u^{\infty}} p_u(S) \cdot w_u \big( 1 - f(\ell_u^{\infty}) \big)
        ~,
    \]
    to the expectation in the lemma. 

    Similarly, all vertices $v \in S$, except for $p^{-1} \big( \ell_u^\infty - \theta_u \big)$ of them, are still matched to offline copies with the value at least $w_u p (1-f(\ell_u^\infty))$ at the time of the matches when $0 \le \theta_u < \ell_u^\infty$.
    This contributes:
    \[
        \textstyle
        \int_0^{\ell_u^{\infty}} e^{-\theta_u} \big( p_u(S) - (\ell_u^{\infty}-\theta_u) \big)^+ d\theta_u \cdot w_u \big( 1 - f(\ell_u^{\infty}) \big)
        ~,
    \]
    to the expectation in the lemma. 
    Putting together proves the lemma
\end{proof}

\begin{proof}[Proof of Theorem~\ref{thm:equal-probabilities} (Vertex-weighted Case)]
    By Lemma~\ref{lem:alpha-equal-prob-weighted} and Lemma~\ref{lem:beta-equal-prob-weighted}, we have:
    \begin{align*}
        & \textstyle
        \E \big[ \alpha_u + \sum_{v \in S} \beta_v \big] \\
        & \textstyle\quad
        \ge
        w_u \bigg( \int_0^{\ell_u^\infty} e^{-\theta_u} f(\theta_u) d\theta_u + \left( e^{-\ell_u^{\infty}} p_u(S) + \int_0^{\ell_u^{\infty}} e^{-\theta_u} \big( p_u(S) - (\ell_u^{\infty}-\theta_u) \big)^+ d\theta_u \right) \big( 1 - f(\ell_u^{\infty}) \big) \bigg)
        ~.
    \end{align*}
    Thus, the dual feasibility derives the identical differential inequality to the unweighted case in Section~\ref{sec:equal-probabilities-primal-dual-unweighted}.
\end{proof}

\section{Unequal Probabilities}
\label{sec:unequal-probabilities}

This section presents the main result of the paper, namely, a $0.572$-competitive algorithm for the more general case of vanishing and unequal probabilities (Theorem~\ref{thm:unequal-probabilities}).
It improves the best previous result of $0.534$ by \citet{MehtaWZ/SODA/2015} in this setting and, surprisingly, is better even than the best previous bound of $0.567$ by \citet{MehtaP/FOCS/2012} in the more restricted setting of vanishing and equal probabilities.

We will consider the following sequence of related problems that are easier to design an algorithm for.
All of them consider a bipartite graph together with the success probabilities of the edges as defined in Section~\ref{sec:model}.
Further, we assume that the success probabilities are upper bounded by some sufficiently small parameter $p > 0$.
We shall establish the competitive ratio with inequalities of the form:
\begin{equation}
    \label{eqn:unequal-competitive}
    \alg \ge \Gamma \cdot \opt - o(1)
    ~,
\end{equation}
where $\Gamma$ is the competitive ratio, and $o(1)$ is a term that tends to $0$ as $p$ tends to $0$.

\begin{itemize}
    \item \textbf{Online Matching with Stochastic Rewards.~}
    We start by recalling the original problem. 
    On the arrival of each online vertex, the algorithm matches it to an offline neighbor which is currently unsuccessful. 
    Then, it succeeds with the probability associated with the edge. 
    \item \textbf{Online (Integral) Matching with Stochastic Budgets.~}
    At the beginning, each offline vertex $u$ independently draws a budget $\theta_u$ from the exponential distribution with mean $1$. 
    On the arrival of an online vertex $v$, the algorithm may match it to any offline neighbor and collect a gain which equals the success probability of the edge.
    However, the gain due to each offline vertex $u$ is capped by its budget $\theta_u$, i.e., it is equal to the sum of success probabilities of the edges matched to $u$, denoted as its load, or $\theta_u$, whichever is smaller.
    Further, $\theta_u$ is unknown to the algorithm until the moment that the total gain due to $u$ exceeds the budget.
    \item \textbf{Online (Fractional) Matching with Stochastic Budgets.~}
    This is the fractional version of the previous problem.
    The only difference is that on the arrival of each online vertex, we may fractionally matched it to multiple offline vertices, provided the total mass does not exceed $1$.
    The increases in the loads of offline vertices are adjusted accordingly, i.e., if a $0.5$ fraction of an online vertex $v$ is matched to $u$, $u$'s load increases by $0.5 p_{uv}$.
\end{itemize}

The reduction from the original problem to the second one formalizes the alternative viewpoint proposed by \citet{MehtaP/FOCS/2012} in the setting of vanishing and unequal probabilities.
The reduction from the second one to the third further simplifies the problem by allowing fractional matching.
Section~\ref{sec:reductions} presents the reductions and shows that it suffices to design a deterministic fractional algorithm for the last problem.

Then, we introduce such a deterministic algorithm in Section~\ref{sec:unequal-probabilities-fractional-algorithm}.
Finally, we explain in Section~\ref{sec:unequal-probabilities-alternating-path} why the alternating path argument used in the case of equal probabilities fails to be generalized.

\subsection{Reductions}
\label{sec:reductions}

In this subsection, we introduce the reductions between the problems and show that it suffices to design a competitive online algorithm for the third problem, i.e., online (fractional) matching with stochastic budgets.

\begin{lemma}
    \label{lem:budget-to-reward}
    Given any online algorithm $A$ that is $\Gamma$-competitive for online (integral) matching with stochastic budgets in the sense of Eqn.~\eqref{eqn:unequal-competitive}, there is an online algorithm $A'$ that is $\Gamma$-competitive for online matching with stochastic rewards in the sense of Eqn.~\eqref{eqn:unequal-competitive}.
\end{lemma}

\begin{lemma}
    \label{lem:fractional-to-integral}
    Given any deterministic online algorithm $A$ that is $\Gamma$-competitive for online (fractional) matching with stochastic budgets in the sense of Eqn.~\eqref{eqn:unequal-competitive}, there is a randomized online algorithm $A'$ that is $\Gamma$-competitive for online (integral) matching with stochastic budgets in the sense of Eqn.~\eqref{eqn:unequal-competitive}.
\end{lemma}

\subsubsection{Proof of Lemma~\ref{lem:budget-to-reward}}

We will construct algorithm $A'$ as follows, using algorithm $A$ as a blackbox.

\begin{enumerate}
    \item On the arrival of each online vertex $v$, algorithm $A'$ matches it to the same offline vertex that algorithm $A$ chooses.
    \item Suppose an offline vertex $u$ succeeds after an edge $(u, v)$ is matched to it in algorithm $A'$.
    Further suppose that before this edge $u$'s load is $\ell_u$.
    Then, draw budget $\theta_u$ from the exponential distribution with mean $1$ conditioned on the value falls into the interval $[\ell_u, \ell_u+p_{uv})$.
    Use it in algorithm $A$.
\end{enumerate}

\begin{lemma}
    [Lemma 2 of \citet{MehtaP/FOCS/2012}]
    \label{lem:success-to-load}
    Given any algorithm for online matching with stochastic rewards, and any offline vertex $u$, the probability that $u$ is successfully matched is equal to the expectation of its load, i.e., the sum of success probabilities of the edges that are matched to $u$ by the algorithm.
\end{lemma}

\begin{lemma}
    For any given bipartite graph $G = (L, R, E)$ and any success probabilities associated with the edges, the expectations of the weighted sum of loads of offline vertices given by $A'$ and $A$ differ by at most $O\big(\sqrt{p} n \sum_{u \in L} w_u \big)$.
\end{lemma}

\begin{proof}
    Let us assume $L = \{1, 2, \dots, n\}$ for simplicity of expositions.
    We proceed with a hybrid argument. 
    Let $A_0$ be algorithm $A$, where the budgets are drawn independently from the exponential distribution with mean $1$.
    Let $A_i$ be an algorithm obtained by realizing the budgets of the first $i$ offline vertices in the way specified in algorithm $A'$.
    Then, $A_n$ is algorithm $A'$.
    It suffices to show that, for any $1 \le u \le n$, the expectations of the sum of weighted loads of offline vertices given by $A_{i-1}$ and $A_i$ differ by at most $O\big(\sqrt{p} \sum_{u \in L} w_u \big)$.
    To do so, it suffices to show that the realizations of $i$'s budget in the two algorithms have a total variation distance at most $O\big(\sqrt{p}\big)$.
    
    Let us first assume the budget of $i$, i.e., the offline vertex that the two algorithms differ in, has a load of $+\infty$.
    In this case, both algorithms realize the same matching.
    Suppose $v_1, v_2, \dots, v_N$ are the online vertices matched to $u$ by the algorithms, with success probabilities $p_{uv_1}, p_{uv_2}, \dots, p_{uv_N}$.
    Then, in the actual realization of the budget in algorithm $A_{i-1}$, for any $1 \le k \le N$, the probability that the algorithm realizes a budget between $\sum_{j=1}^{k-1} p_{uv_j}$ and $\sum_{j=1}^k p_{uv_j}$ is:
    \[
        \int_{\sum_{j=1}^{k-1} p_{uv_j}}^{\sum_{j=1}^k p_{uv_j}} \exp\big(-z\big) dz 
        = \exp\bigg(\sum_{j=1}^{k-1} p_{uv_j}\bigg) \bigg( 1 - \exp\big(-p_{uv_k}\big) \bigg)
        ~.
    \]
    
    On the other hand, in the actual realization of the budget in algorithm $A_i$, for any $1 \le k \le N$, the probability that the algorithm realizes a budget between $\sum_{j=1}^{k-1} p_{uv_j}$ and $\sum_{j=1}^k p_{uv_j}$ is:
    \[
        \prod_{j=1}^{k-1} \big(1 - p_{uv_j} \big) p_{uv_k}
        ~.
    \]
    
    First, note that:
    \[
        p_{uv_k} \ge 1 - \exp\big(-p_{uv_k}\big)
        ~,
    \]
    and
    \[
        1 - \exp \big(-p_{uv_k}\big) \ge 1 - \big( 1 - p_{uv_k} + p_{uv_k}^2 \big) = p_{uv_k} \big( 1 - p_{uv_k} \big) \ge p_{uv_k} (1 - p)
        ~.
    \]
    
    Further, we also have:
    \[
        \exp\bigg(\sum_{j=1}^{k-1} p_{uv_j}\bigg) \ge \prod_{j=1}^{k-1} \big( 1 - p_{uv_j} \big)
        ~,
    \]
    and
    \begin{align*}
        \exp\bigg(\sum_{j=1}^{k-1} p_{uv_j}\bigg)
        &
        \le \prod_{j=1}^{k-1} \big( 1 - p_{uv_j} + p_{uv_j}^2 \big) \\
        &
        \le \prod_{j=1}^{k-1} \big( 1 - p_{uv_j} \big) \big( 1 + 2 p_{uv_j}^2 \big) \\
        & 
        \le \prod_{j=1}^{k-1} \big( 1 - p_{uv_j} \big) \exp \big( 2 p_{uv_j}^2 \big) \\
        & 
        = \bigg( \prod_{j=1}^{k-1} \big( 1 - p_{uv_j} \big) \bigg) \exp \bigg( 2 \sum_{j=1}^{k-1} p_{uv_j}^2 \bigg) \\
        &
        \le \bigg( \prod_{j=1}^{k-1} \big( 1 - p_{uv_j} \big) \bigg) \exp \bigg( 2p \sum_{j=1}^{k-1} p_{uv_j} \bigg)
    \end{align*}
    
    If $\sum_{j=1}^{k-1} p_{uv_j}$ is at most $\frac{1}{\sqrt{p}}$, we have (by $\exp(x) \le 1 + 2x$ for sufficiently small $x > 0$):
    \[
        \exp \bigg( 2p \sum_{j=1}^{k-1} p_{uv_j} \bigg) \le \exp \big( 2 \sqrt{p} \big) \le 1 + 4 \sqrt{p}
        ~.
    \]
    
    Putting together the above inequalities, we get that:
    \[
        \big( 1 + 4 \sqrt{p} \big) \prod_{j=1}^{k-1} \big(1 - p_{uv_j} \big) p_{uv_k}\ge \int_{\sum_{j=1}^{k-1} p_{uv_j}}^{\sum_{j=1}^k p_{uv_j}} \exp\big(-z\big) dz \ge \big(1 - p\big) \prod_{j=1}^{k-1} \big(1 - p_{uv_j} \big) p_{uv_k}
    \]
    
    If $\sum_{j=1}^{k-1} p_{uv_j}$ is larger than $\frac{1}{\sqrt{p}}$, on the other hand, the probability of it happening is small in the first place.
    The total probability mass of such events in algorithm $A_{i-1}$ is at most:
    \[
        \int_{\frac{1}{\sqrt{p}}}^{+\infty} \exp\big(-z\big) dz = \exp \bigg( - \frac{1}{\sqrt{p}} \bigg) \le \sqrt{p}
        ~.
    \]
    
    Combining the two cases shows that the total variation distance between the realizations of $i$'s budget in $A_i$ and $A_{i-1}$ is at most $O\big(\sqrt{p}\big)$.
\end{proof}

\subsubsection{Proof of Lemma~\ref{lem:fractional-to-integral}}

For better exposition of the reduction, we introduce another intermediate problem as follows:
\begin{itemize}
    \item \textbf{Online (Fractional) Matching with $\delta$-Enhanced Stochastic Budgets:~} 
    In this problem, the budget $\theta_u$ of each offline vertex $u$ is equal to $\delta$ plus a random variable independently drawn from the exponential distribution with mean $1$.
\end{itemize}

In particular, we will let $\delta = O \big( \sqrt[3]{p} \log n \big)$ for technical reasons in the analysis.
We start by proving a reduction to this intermediate problem. 

\begin{lemma}
    \label{lem:enhanded-fractional-to-integral}
    Given any deterministic online algorithm $A$ that is $\Gamma$-competitive for online (fractional) matching with $\delta$-enhanced stochastic budgets in the sense of Eqn.~\eqref{eqn:unequal-competitive}, for $\delta = O \big( \sqrt[3]{p} \log n \big)$, there is a randomized online algorithm $A'$ that is $\Gamma$-competitive for online (integral) matching with stochastic budgets in the sense of Eqn.~\eqref{eqn:unequal-competitive}.
\end{lemma}

\begin{proof}
    We will construct algorithm $A'$ as follows:
    \begin{enumerate}
        \item On the arrival of each online vertex $v$:
        \begin{enumerate}
            \item Let $x^A_{uv}$'s, $u \in U$, be the fractional allocation given by algorithm $A$. 
            \item $A'$ matches $v$ to a neighbor with randomness such that it is matched to vertex $u$ with probability $x^A_{uv}$.
        \end{enumerate}
        \item If an offline vertex $u$'s load $\ell_u^{A'}$ exceeds its budget and, thus, the budget $\theta_u^{A'}$ is observed in algorithm $A'$, let its budget in $A$ be $\theta_u^A = \delta + \theta_u^{A'}$.
        \item If at any point there is an offline vertex $u$ whose load $\ell_u^{A'}$ in algorithm $A'$ is less than its load $\ell_u^A$ in algorithm $A$ by more than $\delta$, we say that algorithm $A'$ has failed (to simulate algorithm $A$) since it could be that the budget $\theta_u^A$ is already smaller than the load when it is realized as in the previous step. 
    \end{enumerate}
    
    The weighted sum of loads of algorithm $A'$ is equal to the load of algorithm $A$, modulo a difference of at most $\delta w_u$ per offline vertex $u \in L$ due to the discrepancy in its budget.
    In total, this may result in at most $O(\delta \sum_{u \in L} w_u)$ which is diminishing as $p$ tends to $0$.
    Hence, it remains to bound the losses due to the cases when algorithm $A'$ fails to simulate algorithm $A$.
    
    We first bound the losses when the maximum budget among offline vertices is large. 
    Concretely, let $\theta_{\max} = \max_{u \in U} \theta^{A'}_u$ denote the realized maximum budget in algorithm $A'$.
    Let $F_{\max}$ denote the cdf of this random variable.
    Let $\bar{\theta} = \Theta \big( \log n / \sqrt[3]{p} \big) \gg 1$ be a sufficiently large threshold (since $p$ is sufficiently small).
    We first consider the cases when $\theta_{\max} > \bar{\theta}$.
    Even if the algorithm $A'$ always fails (to simulate algorithm $A$) in these cases, the losses can be bounded.
    First, it is bounded by the total gain of algorithm $A$ in these cases, which is upper bounded by the weighted sum of budgets of offline vertices in algorithm $A$, which in turn is upper bounded by $\sum_{u \in L} w_u$, times the maximum budget $\theta_{\max}$ plus $\delta$ as defined in the above simulation.
    The latter is bounded by:
    \[
        \int_{\bar{\theta}}^{+\infty} \big( \delta + \theta_{\max} \big)
        dF_{\max}(\theta_{\max})
        =
        \big( 1 - F_{\max}(\bar{\theta}) \big) \big( \delta + \bar{\theta} \big) + \int_{\bar{\theta}}^{+\infty} \big( 1 - F_{\max}(\theta_{\max}) \big) d \theta_{\max}
        ~.
    \]
    
    Since each $\theta_u$ is independently drawn from the exponential distribution with mean $1$, for any $\theta_{\max} \ge 0$, the probability that it is larger than $\theta_{\max}$ is $\exp \big( - \theta_{\max} \big)$.
    By union bound, we have:
    \[
        1 - F_{\max}(\theta_{\max}) \le n \exp \big( - \theta_{\max} \big)
        ~.
    \]
    
    Hence, the above losses in the case when $\theta_{\max}$ is large is bounded by $\sum_{u \in L} w_u$ times at most:
    \begin{align*}
        n \big( \delta + \bar{\theta} \big) \exp \big( - \bar{\theta} \big) + \int_{\bar{\theta}}^{+\infty} \exp \big( - \theta_{\max} \big) n d \theta_{\max}
        & 
        =
        n \big( 1 + \delta + \bar{\theta} \big) \exp \big( - \bar{\theta} \big) \\
        &
        < 2 n \bar{\theta} \exp \big( - \bar{\theta} \big) \\[1ex]
        &
        < \sqrt[3]{p} n 
        ~.
    \end{align*}
    
    Next, we account for the losses when the realized budgets are at most $\bar{\theta}$.
    Fix any realization of such budgets $\theta_u^{A'}$'s in algorithm $A'$ and, thus, the budgets $\theta_u^A$'s in algorithm $A$ as well as the fractional matching chosen by $A$. 
    The remaining randomness comes from the matching decisions of algorithm $A'$, which are an independent rounding of the fractional matching by algorithm $A$.
    Fix any offline vertex $u$, we claim that, with high probability, $u$'s load in algorithm $A'$ is at least that in $A$ less $\delta$ at all time.
    This follows from a maximal-Bernstein-style theorem which we state as Lemma~\ref{lem:maximal-bernstein}. For any online vertex $v$, let $X_v = \big( x^A_{uv} - x^{A'}_{uv} \big) \cdot p_{uv}$.
    Then, $X_v$'s are independent random variables with zero means and $|X_v| \le p$ for all $v \in V$.
    Further, we have:
    \begin{align*}
        \sum_{v \in V} \E \big[ X_v^2 \big] 
        & 
        = \sum_{v \in V} \left( \E \big[ \big(x^{A'}_{uv}\big)^2 \big] - \big( x^A_{uv} \big)^2 \right) \cdot p_{uv}^2 \\
        &
        \le \sum_{v \in V} \E \big[ \big(x^{A'}_{uv}\big)^2 \big] \cdot p_{uv}^2 \\
        &
        \le p \sum_{v \in V} \E \big[ x^{A'}_{uv} \cdot p_{uv} \big] \\
        &
        = p \sum_{v \in V} x^{A}_{uv} \cdot p_{uv} \\
        &
        \le p \theta^A_u \\[2ex]
        & 
        \le p \big( \delta + \bar{\theta} \big) \\[2ex]
        & 
        \le 2 p \bar{\theta}
        ~.
    \end{align*}
    
    We w.l.o.g.\ assume that the online vertices are labeled by the order of their arrival. 
    By Lemma~\ref{lem:maximal-bernstein}, we get that with high probability:
    \[
        \max_{1 \le n \le N} \sum_{i=1}^n X_i < O \left( \sqrt{p \bar{\theta} \log n} \right) = O \big( \sqrt[3]{p} \log n \big) \le \delta
        ~.
    \]

    By the definition of $X_v$'s, it means that the load of $u$ in algorithm $A$ is at most that in algorithm $A'$ plus $\delta$.
    Further, by union bound, this holds for all offline vertices with high probability.
    Finally, the losses when it fails to happen is at most the total gain of algorithm $A$, which is upper bounded by the sum of the budgets of offline vertices, which is at most $\bar{\theta} \sum_{u \in L} w_u = O \big( \sum_{u \in L} w_u \log n / \sqrt[3]{p} \big)$.
    
    Putting together the two cases finishes the proof of the lemma.
\end{proof}

Finally, we explain the connection between the intermediate problem and the original fractional problem, i.e., online (fraction) matching with stochastic budgets.

\begin{lemma}
    For any algorithm $A$, the expectation of the sum of loads of offline vertices when it is run on online (fraction) matching with stochastic budgets is at least what it gets on online (fraction) matching with $\delta$-enhanced stochastic budgets, less $O \big( \delta n \log (1/\delta) \sum_{u \in L} w_u \big)$.
\end{lemma}

\begin{proof}
    The total variation distance between the distribution of budgets without enhancements (i.e., drawn independently from the exponential distribution with mean $1$), and that with $\delta$-enhanced budget is at most $O(\delta n)$, where each coordinate contributes at most $O(\delta)$.
    It suffices to bound the difference due to the mismatched probability mass.
    
    Suppose for all mismatched probability mass, the realized budgets are at most $O \big( \log(1/\delta) \big)$.
    Then, the total difference is upper bounded by the gain of the algorithm in the problem with $\delta$-enhanced budgets, which in turn is upper bounded by the sum of the budgets of offline vertices.
    Hence, it is upper bounded by:
    \[
        O \big( \delta n \big) \cdot \big( \delta + O \big( \log(1/\delta) \big) \big) \cdot \sum_{u \in L} w_u 
        \le 
        O \bigg( \delta n \log (1/\delta) \sum_{u \in L} w_u \bigg)
        ~.
    \]
    
    It remains to bound the losses in the tail, i.e., when the maximum budget is more than $O \big( \log(1/\delta) \big)$.
    This part is similar to its counter part in the proof of Lemma~\ref{lem:enhanded-fractional-to-integral}.
    We include it for completeness.
    
    Recall that the budgets $\theta_u$'s in the problem with $\delta$-enhanced budgets are equal to $\delta$ plus some random variable drawn independently from the exponential distribution with mean $1$. 
    Let $\tilde{\theta}_u$'s denote these random variables.
    Let $\theta_{\max} = \max_{u \in U} \theta_u$ denote the realized maximum budget in algorithm $A$ with $\delta$-enhanced budgets.
    Let $\tilde{\theta}_{\max} = \max_{u \in U} \tilde{\theta}_u$.
    By definition, we have $\tilde{\theta}_{\max} = \theta_{\max} - \delta$.
    Let $\tilde{F}_{\max}$ denote the cdf of this random variable.
    
    Let $\bar{\theta} = c \cdot \log(1/\delta) \gg 1$ for some sufficiently large $c > 0$.
    The losses from the cases when the maximum budget $\theta_{\max}$ is at least $\delta + \bar{\theta}$, i.e., when $\tilde{\theta}_{\max}$ is at least $\bar{\theta}$, is upper bounded by the total gain of algorithm $A$ with enhanced budgets, which is upper bounded by the weighted sum of budgets of offline vertices in algorithm $A$, which in turn is upper bounded by $\sum_{u \in L} w_u$ times the maximum budget $\theta_{\max}$ plus $\delta$.
    The latter is:
    \[
        \int_{\bar{\theta}}^{+\infty} \big( \delta + \tilde{\theta}_{\max} \big)
        d\tilde{F}_{\max}(\tilde{\theta}_{\max})
        =
        \big( 1 - \tilde{F}_{\max}(\bar{\theta}) \big) \big( \delta + \bar{\theta} \big) + \int_{\bar{\theta}}^{+\infty} \big( 1 - \tilde{F}_{\max}(\tilde{\theta}_{\max}) \big) d \tilde{\theta}_{\max}
        ~.
    \]
    
    Since each $\tilde{\theta}_u$ is independently drawn from the exponential distribution with mean $1$, for any $\tilde{\theta}_{\max} \ge 0$, the probability that it is larger than $\tilde{\theta}_{\max}$ is $\exp \big( - \theta_{\max} \big)$.
    By union bound, we have:
    \[
        1 - \tilde{F}_{\max}(\tilde{\theta}_{\max}) \le n \exp \big( - \tilde{\theta}_{\max} \big)
        ~.
    \]
    
    Hence, the above losses in the case when $\theta_{\max}$ is large and, thus, $\tilde{\theta}_{\max}$ is large, is upper bounded by $\sum_{u \in L} w_u$ times at most: 
    \begin{align*}
        n \big( \delta + \bar{\theta} \big) \exp \big( - \bar{\theta} \big) + \int_{\bar{\theta}}^{+\infty} \exp \big( - \tilde{\theta}_{\max} \big) n d \tilde{\theta}_{\max}
        & 
        =
        n \big( 1 + \delta + \bar{\theta} \big) \exp \big( - \bar{\theta} \big) \\
        &
        < 2 n \bar{\theta} \exp \big( - \bar{\theta} \big) \\[1ex]
        &
        < O \big( \delta n \log (1/\delta) \big)
        ~.
    \end{align*}
    
    Putting together proves the lemma.
\end{proof}

\subsection{Online (Fractional) Matching with Stochastic Budgets}
\label{sec:unequal-probabilities-fractional-algorithm}

This section completes the proof of Theorem~\ref{thm:unequal-probabilities} by giving a deterministic online algorithm for the online (fractional) vertex-weighted matching with stochastic budgets, and the corresponding competitive analysis for the desired ratio. 

\subsubsection{Algorithm}

The algorithm is driven by the randomized online primal dual framework.
In particular, recall that gain sharing rule in the case of equal probabilities.
When the algorithm picks an edge $(u, v)$, the gain of $p$ from the edge is split among $u$ and $v$ according to the current load $\ell_u$ and the weight $w_u$ of the offline vertex $u$ and some non-decreasing function $f$:
the dual variable of $u$ increases by $p w_u f(\ell_u)$;
the dual variable of $v$ is set to be $p w_u \big( 1 - f(\ell_u) \big)$.
Stochastic Balance and Algorithm~\ref{alg:equal-probabilities-weighted} can both be viewed as matching the online vertex $v$ to the offline neighbor that gives $v$ the largest share of the gain.
Following the same reason, our algorithm continuously matches infinitesimal fraction of $v$ to the offline neighbor with the largest $p_{uv} w_u \big( 1 - f(\ell_u) \big)$, which equals the $v$'s sharing of the gain (per unit of match).
Algorithm~\ref{alg:unequal-probabilities} presents the formal definition of the algorithm.


\begin{algorithm}
    \caption{A deterministic online algorithm for online (fractional) matching with stochastic budgets, parameterized by a non-decreasing function $f : \R^+ \mapsto \R^+$}
    \label{alg:unequal-probabilities}
    \For{each online vertex $v \in V$}{
        initialize $x_{uv}=0$ for all offline neighbors  $u \in N_v$\\
        \For{$t$ increasing from $0$ to $1$ continuously}{
            let $u^*$ be the unsuccessful neighbor with the largest $p_{uv} w_u \big( 1 - f(\ell_u) \big)$\\
            increase $x_{u^*v}$ with rate $1$, i.e., for any $v$'s neighbor $u$:
            \[
                \frac{dx_{uv}}{dt} = \begin{cases}
                    1 & u = u^* \\
                    0 & u \ne u^*
                \end{cases}
            \]
        }
    }
\end{algorithm}

Alternatively, the algorithm can be interpreted as follows.
On the arrival of each online vertex $v \in V$, let $\widetilde{N}_v$ denote the set of unsuccessful neighbors of $v$.
If $\widetilde{N}_v$ is not empty, for any threshold $0 \le \tau \le p \cdot \text{max}_{u \in U} w_u$, define $x_{uv}(\tau)$, $u \in \widetilde{N}_v^*$ as follows:  
\[
    x_{uv}(\tau) = \sup \big\{ 0 \le x \le 1 : p_{uv} w_u \big( 1 - f(\ell_u + p_{uv} x) \big) \ge \tau \text{ or } \ell_u + p_{uv} x \le \theta_u \big\}
    ~,
\]

That is, we let $x_{uv}(\tau)$ be the largest possible fraction of edge $(u, v)$ to be matched, so long as the marginal gain per unit of match is at least the threshold $\tau$, and the load of $u$ does not exceed its threshold.
Let $\tau$ be such that either $\tau = 0$ and $\sum_{u \in \widetilde{N}_v} x_{uv}(\tau) < 1$, or $\tau > 0$ and $\sum_{u \in \widetilde{N}_v} x_{uv}(\tau) = 1$.
Let $x_{uv} = x_{uv}(\tau)$.
Such a $\tau$ exists because $x_{uv}(\tau)$ is non-increasing in $\tau$, and $x_{uv}(+\infty) = 0$.


\subsubsection{Dual Assignments}
\label{sec:unequal-probabilities-dual}

We now explain how to maintain a dual assignment.
It is similar to that in the case of equal probabilities, except that the gains from different edges may be different, and that the matching process is fractional.
Recall that $\ell_u$ denotes the current load of an offline vertex $u$.
For some non-decreasing function $f : \R^+ \mapsto \R^+$ to be determined in the analysis, we maintain the dual assignment as follows.

\begin{enumerate}
    \item Initially, let $\alpha_u = 0$ for all $u \in U$.
    \item On the arrival of an online vertex $v \in V$ and, thus, the corresponding dual variable $\beta_v$:
    \begin{enumerate}
        \item Initialize $\beta_v = 0$.
        \item As $t$ increases continuously from $0$ to $1$ in the algorithm, increase $\alpha_{u^*}$ and $\beta_v$ such that:
            \begin{align*}
                \frac{d \alpha_{u^*}}{dt} & = p_{u^*v} w_{u^*} f(\ell_{u^*}) \\
                \frac{d \beta_v}{dt} & = p_{u^*v} w_{u^*} \big( 1 - f(\ell_{u^*}) \big)
            \end{align*}
    \end{enumerate}
\end{enumerate}

Let $F(\ell) = \int_0^\ell f(z) dz$ denote the integration of $f$.
Note that for any offline vertex $u$, the dual variable $\alpha_u$ always increases at a rate proportional to the rate that the load increases, multiply by $ w_{u} f(\ell_u)$.
We have $\alpha_u = w_u F(\ell_u)$ by definition.


\subsubsection{Invariants}

Next, fix any offline vertex $u \in U$, and any subset $S$ of its online neighbors, and fix any thresholds $\vec{\theta}_{-u}$ of the offline vertices other than $u$.
Further, the set of online vertices naturally partition the matching process into $|V|+1$ moments, i.e., before the arrival of any vertex, after the arrival of the first vertex, after the arrival of the second vertex, and so on.
This subsection establishes some basic invariants of the matching given by the above algorithm when the threshold $\theta_u$ of vertex $u$ changes by comparing the matchings at any given moment before and after the change.

\begin{lemma}
    \label{lem:load-invariant}
    For any moment of the matching process, and any offline vertex $u' \ne u \in U$, as $\theta_u$ decreases, the load $\ell_{u'}$ of vertex $u'$ at the given moment weakly increases.
\end{lemma}

\begin{proof}
    We will refer to the matching before $\theta_u$ decreases as the \emph{old} instance, and the instance after $\theta_u$ decreases as the \emph{new} instance.
    The proof proceeds by induction on the moments from the beginning of the matching process to the end. 

    At the beginning, the lemma holds trivially. 
    In fact, it holds trivially up to the moment that $u$ succeeds in the new instance after $\theta_u$ decreases.

    Next, suppose the lemma holds at the moment before the arrival of a vertex $v$.
    We next argue that it continues to hold at the moment after the arrival of $v$.
    It is most convenient to prove it using the alternative interpretation of the algorithm using the threshold $\tau$.

    We use superscripts to distinguish the value of the same variable in the old and new instances.
    For example, $\tau^{\text{old}}$ and $\tau^{\text{new}}$ denote the thresholds in the old and new instances respectively.
    We claim that $\tau^{\text{old}} \ge \tau^{\text{new}}$.
    First, note that for any given $\tau$, and any given offline vertex $u$, the value of $x_{uv}^{\text{new}}(\tau)$ is weakly smaller than $x_{uv}^{\text{old}}$ because $\ell_u^{\text{new}}$ is weakly larger than $\ell_u^{\text{old}}$ at this moment.
    Further, the weak monotonicity of the loads of offline vertices also imply that the set of unsuccessful neighbors $\widetilde{N}_v^{\text{new}}$ in the new instance is a subset of $\widetilde{N}_v^{\text{old}}$, i.e., that in the old instance.
    Therefore:
    \[
        \textstyle
        \sum_{u \in \widetilde{N}_v^{\text{new}}} x_{uv}^{\text{new}}(\tau^{\text{old}}) \le \sum_{u \in \widetilde{N}_v^{\text{old}}} x_{uv}^{\text{old}}(\tau^{\text{old}}) = 1
        ~.
    \]
    So the claim of $\tau^{\text{old}} \ge \tau^{\text{new}}$ follows by that $x_{uv}(\tau)$ is non-increasing in $\tau$.

    As a corollary, the loads of the offline vertices are still weakly larger in the new instance after the arrival of $v$.
    For any vertex $u$, if it is not a neighbor of $v$, it holds trivially by the induction hypothesis.
    If it succeeds after the arrival of $v$, it also holds trivially since the load is now equal to the threshold.
    Finally, if it remains unsuccessful after the arrival of $v$, its load is determined by the threshold. 
    Since the threshold is weakly smaller in the new instance, the load is weakly larger.
\end{proof}

As a direct corollary, we have the following invariant about the online vertices.

\begin{lemma}
    \label{lem:unequal-probabilities-online-invariant}
    For any online vertex $v \in V$, as $\theta_u$ decreases, the value of the dual variable $\beta_v$ weakly decreases.
\end{lemma}

\begin{proof}
    Since the loads are weakly higher in the new instance by Lemma~\ref{lem:load-invariant}, the dual variable $\beta_v$ increases at a lower rate by definition.
    So the lemma follows.
\end{proof}

\subsubsection{Structural Lemma for Unequal Probabilities}
\label{sec:unequal-probabilities-structural}

Building on the invariants established in the previous subsection, we now explain the main structural lemma that the online primal dual analysis utilizes in the case of unequal probabilities.
To formally describe the structural lemma, let us first introduce the following notation.
Given any thresholds $\vec{\theta}$, let $\beta_v(\vec{\theta})$ denote the value of $\beta_v$ when the thresholds are $\vec{\theta}$.
We abuse notation and let $(\infty, \vec{\theta}_{-u})$ denote the case when the threshold $\theta_u$ is sufficiently large such that vertex $u$ never succeeds.

\begin{lemma}
    \label{lem:unequal-probabilities-structural}
    For any thresholds $\vec{\theta}_{-u}$ of the offline vertices other than $u$ and the corresponding load $\ell_u^\infty$ when $\theta_u = \infty$, for any $\theta_u \le \ell_u^\infty$, and for any subset $S$ of the offline vertices, we have:
    \[
        \textstyle
        \sum_{v \in S} \beta_v \big( \theta_u, \vec{\theta}_{-u} \big) \geq \sum_{v \in S} \beta_v \big( \infty, \vec{\theta}_{-u} \big) - \int_{\theta_u}^{\ell_u^{\infty}} w_u \big(1 - f(z) \big) dz
        ~.
    \]
\end{lemma}

\begin{proof}
    It suffices to show the inequality for $S = V$.
    Then, the general case follows by subtracting the contributions of $\beta_v$ for $v \notin S$, due to Lemma~\ref{lem:unequal-probabilities-online-invariant}.

    To show the lemma for $S = V$, first note that by the gain sharing rule, for any realization of the thresholds, the primal objective, i.e., $\sum_{u' \in U} w_{u'} \ell_{u'}$, is equal to the dual objective, i.e., $\sum_{u' \in U} \alpha_{u'} + \sum_{v \in V} \beta_v$.
    Further, by that $\alpha_{u'} = w_{u'} F(\ell_{u'}) = \int_0^{\ell_{u'}} w_{u'} f(z) dz$, we have:
    \[
        \textstyle
        \sum_{v \in V} \beta_v = \sum_{u' \in U} \int_0^{\ell_{u'}} w_{u'} \big( 1 - f(z) \big) dz
        ~.
    \]
    
    Comparing the RHS in the instances that correspond to thresholds $\big( \infty, \vec{\theta}_{-u} \big)$ and $\big( \theta_u, \vec{\theta}_{-u} \big)$ proves the lemma, because the loads $\ell_{u'}$ for any offline vertex $u' \ne u$ weakly increase by Lemma~\ref{lem:load-invariant}, and the load $\ell_u$ of vertex $u$ decreases from $\ell_u^\infty$ to $\theta_u$.
\end{proof}

\subsubsection{Randomized Online Primal Dual Analysis}

Note that the dual assignment ensures that the primal and dual objectives are equal at all time.
It suffices to show approximate dual feasibility in expectation, i.e., for any offline vertex $u$ and any subset $S$ of its online neighbors:
\[
    \textstyle
    \E \big[ \alpha_u + \sum_{v\in S} \beta_v \big] \geq \Gamma \cdot p_u(S) \cdot w_u
    ~.
\]

Following the same strategy as in the case of equal probabilities, we fix the realization of the thresholds $\vec{\theta}_{-u}$ of all offline vertices other than $u$, and prove the above inequality over the randomness of the threshold $\theta_u$ of vertex $u$ alone.
Recall that $\ell_u^\infty$ denote the load of $u$ when $\theta_u = \infty$, i.e., when $u$ never succeeds.

Next, we characterize the matching related to $u$ and online vertices in $S$ for different realization of the threshold $\theta_u$ of $u$.
It is similar to its counterpart in the case of equal probabilities, except that we replace the structural lemma with the one for unequal probabilities.
\begin{itemize}
    \item \textbf{Case 1: $\theta_u \ge \ell_u^\infty$.~}
        In this case, the matching is the same as the $\theta_u = \infty$ case since $u$ has enough budget to accommodate all the load matched to it.  
    \begin{itemize}
        \item The load of $u$ equals $\ell_u^\infty$.
        \item The total gain of all online vertices in $S$ is at least $w_u p_u(S) \big( 1-f(\ell_u^{\infty}) \big)$, since at all time $u$ is available.
    \end{itemize}
    \item \textbf{Case 2: $0 \le \theta_u < \ell_u^\infty$.~}
        In this case, $u$ can only accommodate the load up to its threshold. 
        Some online vertices that were matched to $u$ when $\theta_u = \infty$ need to be matched elsewhere.
    \begin{itemize}
        \item The load of $u$ equals $\theta_u$.
        \item The total gain of the online vertices in $S$ is lower bounded by Lemma~\ref{lem:unequal-probabilities-structural}.
    \end{itemize}
\end{itemize}

\paragraph{Contribution from $\alpha_u$.}
This part is a verbatim copy of the counterpart in the case of equal probabilities.
We reiterate the lemma below and omit the proof.

\begin{lemma}
    \label{lem:unequal-probabilities-alpha}
    For any thresholds $\vec{\theta}_{-u}$ of the offline vertices other than $u$ and the corresponding $\ell_u^\infty$:
    \[
        \textstyle
        \E _{\theta_u} \big[ \alpha_u | \vec{\theta}_{-u} \big] \geq \int_0^{\ell_u^\infty} e^{-\theta_u} w_u f(\theta_u) d\theta_u 
        ~.
    \]
\end{lemma}

\paragraph{Contribution from $\beta_v$'s.}
%
We now turn to the contribution from the online vertices. 
Recall that $z^+$ denotes $\max \big\{ z, 0 \big\}$.

\begin{lemma}
    \label{lem:unequal-probabilities-beta}
    For any thresholds $\vec{\theta}_{-u}$ of the offline vertices other than $u$ and the corresponding $\ell_u^\infty$:
    \[
        \textstyle
        \E _{\theta_u} \big[ \sum_{v\in S} \beta_v | \vec{\theta}_{-u} \big]
        \geq w_u \left( e^{-\ell_u^{\infty}} p_u(S) \big( 1-f(\ell_u^{\infty}) \big)
    + \int_0^{\ell_u^{\infty}} \left( p_u(S) \big( 1-f(\ell_u^{\infty}) \big) - \int_{\theta_u}^{\ell_u^{\infty}} \big( 1-f(y) \big) dy \right) ^+ e^{-\theta_u} d\theta_u \right)
        ~.
    \]
\end{lemma}

\begin{proof}
    If $\theta_u \ge \ell_u^{\infty}$, which happens with probability $e^{-\ell_u^\infty}$, the total gain of all online vertices in $S$, i.e., $\sum_{v \in S} \beta_v$, is at least:
    \[
        w_u p_u(S) \big( 1-f(\ell_u^{\infty}) \big)
        ~.
    \]
    
    If $ \theta_u < \ell_u^{\infty}$, by Lemma~\ref{lem:unequal-probabilities-structural} we get that $\sum_{v \in S} \beta_v$ is at least:
    \[
        \textstyle
        w_u \left( p_u(S) \big( 1-f(\ell_u^{\infty}) \big) - \int_{\theta_u}^{\ell_u^{\infty}} \big( 1-f(y) \big) dy \right)
        ~.
    \]
    
    We also have that $\sum_{v \in S} \beta_v$ is non-negative.
    Putting together proves the lemma.
%
\end{proof}

\subsubsection{Proof of Theorem~\ref{thm:unequal-probabilities}}
Combining Lemma~\ref{lem:unequal-probabilities-alpha} and Lemma~\ref{lem:unequal-probabilities-beta}, we get that:
\begin{align*}
    \textstyle
    \E \big[ \alpha_u + \sum_{v \in S} \beta_v \big] & 
    \textstyle
    \ge w_u \bigg( \int_0^{\ell_u^\infty} e^{-\theta_u} f(\theta_u) d\theta_u + e^{-\ell_u^{\infty}} p_u(S) \big( 1-f(\ell_u^{\infty}) \big) \\
    & \qquad \textstyle
    + \int_0^{\ell_u^{\infty}} \left( p_u(S) \big( 1-f(\ell_u^{\infty}) \big) - \int_{\theta_u}^{\ell_u^{\infty}} \big( 1-f(y) \big) dy \right) ^+ e^{-\theta_u} d\theta_u \bigg)
    ~.
\end{align*}

Hence, to show a competitive ratio $\Gamma$, it suffices to find non-decreasing function $f$ that satisfy the following set of differential inequalities.
For ease of notations, we write $\ell_u^\infty$ as $\ell$, $p_u(S)$ as $p$, and $\theta_u$ as $\theta$.
For any $\ell \ge 0$ and any $0 \le p \le 1$, we need:
\begin{equation}
    \label{eqn:de-unequal}
    \textstyle
    \int_0^{\ell} e^{-\theta} f(\theta) d\theta + e^{-\ell} p \big( 1-f(\ell) \big) 
    + \int_0^{\ell} \left( p \big( 1-f(\ell) \big) - \int_{\theta}^{\ell} \big( 1-f(y) \big) dy \right) ^+ e^{-\theta} d\theta 
    \geq \Gamma \cdot p
    ~.
\end{equation}

This is not a standard differential equation that can be solved by standard method, to the best of our knowledge.
In fact, we do not know how to optimize the competitive ratio exactly.
Below, we explain how to simplify it so that we can find a good enough solution.

\paragraph{Left Boundary Condition.}
We start with the left boundary condition of $f$ at zero.

\begin{lemma}
    \label{lem:unequal-probabilities-de-left-boundary}
    For solution $f$ to the differential inequalities, we have:
    \[
        f(0) \le 1 - \Gamma
        ~.
    \]
\end{lemma}

\begin{proof}
    It follows by letting $\ell = 0$ and, e.g., $p = 1$, in Eqn.~\eqref{eqn:de-unequal}.
\end{proof}

In the following discussion, we focus on functions such that the left boundary condition holds with equality, i.e., $f(0) = 1 - \Gamma$.

\paragraph{Unit Mass of Online Neighbors.}
Then, we show that the worst-case scenarios are when there is a unit mass of online neighbors, i.e., when $p_u(S) = 1$.

\begin{lemma}
    If $f$ is a non-decreasing function such that $f(0) = 1 - \Gamma$ and Eqn.~\eqref{eqn:de-unequal} holds for any $\ell \ge 0$ and $p = 1$.
    Then, $f$ satisfies Eqn.~\eqref{eqn:de-unequal} for any $\ell \ge 0$ and any $0 \le p \le 1$.
\end{lemma}

\begin{proof}
    Let function $g : \R^+ \mapsto \R^+$ be such that for any $\ell \ge 0$, we have either that:
    \[
        \textstyle
        1-f(\ell) = \int_{g(\ell)}^{\ell} \big( 1-f(y) \big) dy
        ~,
    \]
    or that $g(\ell) = 0$ and:
    \[
        \textstyle
        1 - f(\ell) > \int_0^{\ell} \big( 1-f(y) \big) dy
        ~.
    \]

    Then, for any $\ell \ge 0$, we have:
    \[
        \textstyle
        \int_0^{\ell} \left( \big( 1-f(\ell) \big) - \int_{\theta}^{\ell} \big( 1-f(y) \big) dy \right) ^+ e^{-\theta} d\theta 
        =  
        \int_{g(\ell)}^{\ell} \left( \big( 1-f(\ell) \big) - \int_{\theta}^{\ell} \big( 1-f(y) \big) dy \right) e^{-\theta} d\theta
        ~.
    \]

    By our assumption that function $f$ satisfies Eqn.~\eqref{eqn:de-unequal} when $p = 1$, we get that:
    \begin{equation}
        \label{eqn:de-unequal-unit-mass}
        \textstyle
        \int_0^{\ell} e^{-\theta} f(\theta) d\theta + e^{-\ell} \big( 1-f(\ell) \big) 
        + \int_{g(\ell)}^{\ell} \left( \big( 1-f(\ell) \big) - \int_{\theta}^{\ell} \big( 1-f(y) \big) dy \right) e^{-\theta} d\theta 
        \geq \Gamma
        ~.
    \end{equation}

    To show that Eqn.~\eqref{eqn:de-unequal} holds more generally for any $0 \le p \le 1$, note that:
    \begin{align*}
        & \textstyle
        \int_0^{\ell} e^{-\theta} f(\theta) d\theta + e^{-\ell} p \big( 1-f(\ell) \big) + \int_0^{\ell} \left( p \big( 1-f(\ell) \big) - \int_{\theta}^{\ell} \big( 1-f(y) \big) dy \right)^+ e^{-\theta} d\theta \\
        & \textstyle \qquad
        \ge \int_0^{\ell} e^{-\theta} f(\theta) d\theta + e^{-\ell} p \big( 1-f(\ell) \big) 
        + \int_{g(\ell)}^{\ell} \left( p \big( 1-f(\ell) \big) - \int_{\theta}^{\ell} \big( 1-f(y) \big) dy \right) e^{-\theta} d\theta
        ~.
    \end{align*}

    Hence, it suffices to show that:
    \[
        \textstyle
        \int_0^{\ell} e^{-\theta} f(\theta) d\theta + e^{-\ell} p \big( 1-f(\ell) \big) 
        + \int_{g(\ell)}^{\ell} \left( p \big( 1-f(\ell) \big) - \int_{\theta}^{\ell} \big( 1-f(y) \big) dy \right) e^{-\theta} d\theta \ge \Gamma p
        ~.
    \]

    Eqn.~\eqref{eqn:de-unequal-unit-mass} shows that it holds when $p = 1$.
    Consider both sides as a function of $p$.
    It suffices to show that the derivative of the RHS is weakly larger than that of the LHS for any $0 \le p \le 1$.
    The derivative of the RHS is $\Gamma$, while that of the LHS at most: 
    \begin{align*}
        \textstyle
        e^{-\ell} \big( 1 - f(\ell) \big) + \int_{g(\ell)}^\ell \big( 1 - f(\ell) \big) e^{-\theta} d\theta 
        & 
        = e^{-g(\ell)} \big( 1 - f(\ell) \big) \\
        & 
        \le 1 - f(\ell) 
        && \text{($g(\ell) \ge 0$)} \\[1ex]
        &
        \le 1 - f(0) 
        && \text{($f$ is non-decreasing)} \\[1ex]
        & 
        = \Gamma
        && \text{(left boundary condition of $f$)}
        ~.
    \end{align*}

    So the lemma follows.
\end{proof}

Therefore, it suffices to show the following simplified differential inequalities subject to the left boundary condition that $f(0) = 1 - \Gamma$.
For any $\ell \ge 0$, we need:
\begin{equation}
    \label{eqn:de-unequal-simplified}
    \textstyle
    \int_0^{\ell} e^{-\theta} f(\theta) d\theta + e^{-\ell} \big( 1-f(\ell) \big) 
    + \int_0^{\ell} \left( \big( 1-f(\ell) \big) - \int_{\theta}^{\ell} \big( 1-f(y) \big) dy \right) ^+ e^{-\theta} d\theta 
    \geq \Gamma 
    ~.
\end{equation}

\paragraph{Right Boundary Condition.~}
To further simplification, we impose a right boundary condition that it suffices to ensure that $f(\ell) = 1 - \frac{1}{e}$ for some sufficiently large $\ell$.

\begin{lemma}
    \label{lem:unequal-probabilities-de-right-boundary}
    For any $\ell_\max > 0$, suppose $f : [0, \ell_\max] \mapsto \R^+$ is a non-decreasing function such that $f(0) = 1 - \Gamma$, $f(\ell_\max) = 1 - \frac{1}{e}$, and Eqn.~\eqref{eqn:de-unequal-simplified} holds for any $0 \le \ell \le \ell_\max$.
    Then, the extension of $f$ such that $f(\ell) = f(\ell_\max) = 1 - \frac{1}{e}$ for all $\ell \ge \ell_\max$ satisfies Eqn.~\eqref{eqn:de-unequal} for all $\ell \ge 0$.
\end{lemma}

\begin{proof}
    Let $d \ge 0$ be such that:
    \[
        \textstyle
        1 - f(\ell_\max) = \int_{\ell_\max - d}^{\ell_\max} \big( 1-f(y) \big) dy 
        ~.
    \]

    Note that $f$ is non-decreasing implies that $d \le 1$

    For any $\ell \ge \ell_\max$, we have:
    \begin{align*}
        & \textstyle
        \int_0^{\ell} e^{-\theta} f(\theta) d\theta + e^{-\ell} \big( 1-f(\ell) \big) + \int_0^{\ell} \left( \big( 1-f(\ell) \big) - \int_{\theta}^{\ell} \big( 1-f(y) \big) dy \right) ^+ e^{-\theta} d\theta \\
        & \textstyle \qquad
        \ge \int_0^{\ell} e^{-\theta} f(\theta) d\theta + e^{-\ell} \big( 1-f(\ell) \big) + \int_{\ell-d}^{\ell} \left( \big( 1-f(\ell) \big) - \int_{\theta}^{\ell} \big( 1-f(y) \big) dy \right) e^{-\theta} d\theta \\
        & \textstyle \qquad
        = \int_0^{\ell} e^{-\theta} f(\theta) d\theta + e^{-\ell+d} \big( 1-f(\ell) \big) - \int_{\ell-d}^{\ell} \int_{\theta}^{\ell} \big( 1-f(y) \big) dy ~ e^{-\theta} d\theta \\[1ex]
        & \textstyle \qquad
        = \int_0^{\ell} e^{-\theta} f(\theta) d\theta + e^{-\ell+d} \big( 1-f(\ell) \big) - \int_{\ell-d}^{\ell} \int_{\ell-d}^y e^{-\theta} d\theta ~ \big( 1-f(y) \big) dy \\[1ex]
        & \textstyle \qquad 
        = \int_0^{\ell} e^{-\theta} f(\theta) d\theta + e^{-\ell+d} \big( 1-f(\ell) \big) - \int_{\ell-d}^{\ell} \big( e^{-\ell+d} - e^{-y} \big) \big( 1-f(y) \big) dy \\[1ex]
        & \textstyle \qquad 
        = \int_0^{\ell} e^{-\theta} f(\theta) d\theta + e^{-\ell+d} \big( 1-f(\ell) \big) - e^{-\ell+d} \int_{\ell-d}^{\ell} \big( 1-f(y) \big) dy + \int_{\ell-d}^{\ell} e^{-y} \big( 1-f(y) \big) dy 
        ~.
    \end{align*}

    Note that this is at least $\Gamma$ when $\ell = \ell_\max$ by the assumption on $f$ and the definition of $d$.
    It suffices to show that the RHS, as a function of $\ell$, is non-decreasing.
    The derivative of each term on the RHS w.r.t.\ $\ell$ equals (recalling that $f(\ell) = 1 - \frac{1}{e}$ for all $\ell \ge \ell_\max$):
    \begin{eqnarray*}
        \textstyle
        \frac{d}{d\ell} \int_0^{\ell} e^{-\theta} f(\theta) d\theta 
        & = &
        \textstyle
        e^{-\ell} \big( 1 - \frac{1}{e} \big) \\[1ex]
        \textstyle
        \frac{d}{d\ell} e^{-\ell+d} \big( 1-f(\ell) \big) 
        & = & 
        \textstyle
        - e^{-\ell+d-1} \\[1ex]
        \textstyle
        \frac{d}{d\ell} e^{-\ell+d} \int_{\ell-d}^{\ell} \big( 1-f(y) \big) dy
        & = &
        \textstyle
        - e^{-\ell+d} \int_{\ell-d}^{\ell} \big( 1-f(y) \big) dy + e^{-\ell+d} \big( \frac{1}{e} - \big(1 - f(\ell-d) \big) \big) \\[1ex]
        \textstyle
        \frac{d}{d\ell} \int_{\ell-d}^{\ell} e^{-y} \big( 1-f(y) \big) dy 
        & = &
        \textstyle
        e^{-\ell-1} - e^{-\ell+d} \big( 1 - f(\ell-d) \big)
    \end{eqnarray*}

    Putting together shows that the derivative of the RHS w.r.t.\ $\ell$ equals:
    \begin{align*}
        \textstyle
        e^{-\ell} - 2 e^{-\ell+d-1} + e^{-\ell+d} \int_{\ell-d}^{\ell} \big( 1-f(y) \big) dy
        &
        \textstyle
        \ge e^{-\ell} - 2 e^{-\ell+d-1} + e^{-\ell+d} \int_{\ell-d}^{\ell} \frac{1}{e} dy \\
        & 
        \textstyle
        = e^{-\ell} \big( 1 - (2-d) e^{d-1} \big) \\
        & 
        \textstyle
        \ge 0 
    \end{align*}

    The first inequality follows by that $f$ is non-decreasing.
    The second inequality follows by that $1 - (2-d) e^{d-1}$ is decreasing in $0 \le d \le 1$ and that it holds with equality for $d = 1$.
\end{proof}

As a result, we can focus on the following further simplified version of the differential inequalities.
For some appropriately chosen $\ell_\max > 0$, it suffices to find non-decreasing $f$ subject to the boundary conditions $f(0) = 1 - \Gamma$ and $f(\ell_\max) = 1 - \frac{1}{e}$, such that Eqn.~\eqref{eqn:de-unequal-simplified} holds for any $0 \le \ell \le \ell_\max$.

\paragraph{Weaker Bound via an Explicit $f$.}
We first show a weaker competitive ratio of $\frac{1 + e^{-2}}{2} \approx 0.567$ via an explicit construction of $f$.
This nevertheless is better than the previous best bound of $0.534$ by \citet{MehtaWZ/SODA/2015} and is equal to the best previous bound by \citet{MehtaP/FOCS/2012} in the more restricted case of equal probabilities.

To do so, we let $\ell_\max = 1$ and further relax the differential inequalities by dropping the $z^+$ operations.
That is, we solve for a non-decreasing $f$ subject to the boundary conditions that $f(0) = 1 - \Gamma$ and $f(1) = 1 - \frac{1}{e}$ such that for any $0 \le \ell \le 1$, we have:
\[
    \textstyle
    \int_0^{\ell} e^{-\theta} f(\theta) d\theta + e^{-\ell} \big( 1-f(\ell) \big) 
    + \int_0^{\ell} \left( \big( 1-f(\ell) \big) - \int_{\theta}^{\ell} \big( 1-f(y) \big) dy \right) e^{-\theta} d\theta 
    \geq \Gamma 
    ~.
\]

Simplifying the LHS through a sequence of equalities, we get:
\begin{align*}
    & \textstyle
    \int_0^{\ell} e^{-\theta} f(\theta) d\theta + e^{-\ell} \big( 1-f(\ell) \big) + \int_0^{\ell} \left( \big( 1-f(\ell) \big) - \int_{\theta}^{\ell} \big( 1-f(y) \big) dy \right) e^{-\theta} d\theta \\
    & \textstyle \qquad
    = \int_0^{\ell} e^{-\theta} f(\theta) d\theta + e^{-\ell} \big( 1-f(\ell) \big) + \int_0^{\ell} \big( 1-f(\ell) \big) e^{-\theta} d\theta - \int_0^\ell \int_\theta^\ell \big( 1-f(y) \big) dy ~ e^{-\theta} d\theta \\
    & \textstyle \qquad
    = \int_0^{\ell} e^{-\theta} f(\theta) d\theta + \big( 1-f(\ell) \big) - \int_0^\ell \int_\theta^\ell \big( 1-f(y) \big) dy ~ e^{-\theta} d\theta \\
    & \textstyle \qquad
    = \int_0^{\ell} e^{-\theta} f(\theta) d\theta + \big( 1-f(\ell) \big) - \int_0^\ell \int_0^y e^{-\theta} d\theta ~ \big( 1-f(y) \big) dy \\
    & \textstyle \qquad
    = \int_0^{\ell} e^{-\theta} f(\theta) d\theta + \big( 1-f(\ell) \big) - \int_0^\ell \big( 1 - e^{-y} \big) \big( 1-f(y) \big) dy \\
    & \textstyle \qquad
    = \int_0^{\ell} e^{-\theta} f(\theta) d\theta + \big( 1-f(\ell) \big) - \int_0^\ell \big( 1-f(y) \big) dy 
 + \int_0^\ell e^{-y} \big( 1-f(y) \big) dy \\
    & \textstyle \qquad
    = \int_0^\ell e^{-\theta} d\theta + \big( 1-f(\ell) \big) - \int_0^\ell \big( 1-f(y) \big) dy \\
    & \textstyle \qquad
    = 1 - e^{-\ell} + \big( 1-f(\ell) \big) - \int_0^\ell \big( 1-f(y) \big) dy 
    ~.
\end{align*}

Hence, it suffices to find a non-decreasing $f$ subject to the boundary conditions that $f(0) = 1 - \Gamma$ and $f(1) = 1 - \frac{1}{e}$ such that for any $0 \le \ell \le 1$, we have:
\begin{equation}
    \label{eqn:de-unequal-simplest}
    \textstyle
    1 - e^{-\ell} + \big( 1-f(\ell) \big) - \int_0^\ell \big( 1-f(y) \big) dy 
    \geq \Gamma 
    ~.
\end{equation}

Letting the above hold with equality for $0 \le \ell \le 1$, it is a standard differential equation that admits an explicit solution as follows:
%
\begin{equation}
    \label{eqn:f-unequal}
    f(x) = \begin{cases}
            1-\frac{1}{e} & \quad x > 1 \\
            1-\frac{1}{2} e^{-x} -\frac{1}{2e^2} e^x& \quad 0 \leq x \leq 1
        \end{cases}
\end{equation}

It satisfies Eqn.~\eqref{eqn:de-unequal-simplest} with equality for $\Gamma = 1-f(0) = \frac{1}{2}(1+e^{-2}) \approx 0.567$. 

\paragraph{Better Bound via a Computer-assisted Analysis.}
To further optimize the competitive ratio, we rely on a computer-assisted analysis as follows.
Let $\ell_\infty = 2$.
Further, for some appropriate function $g : \R^+ \mapsto \R^+$, relax the differential inequalities so that for any $0 \le \ell \le \ell_\max$, we have:
\begin{equation}
    \label{eqn:de-unequal-relaxed}
    \textstyle
    \int_0^{\ell} e^{-\theta} f(\theta) d\theta + e^{-\ell} \big( 1-f(\ell) \big) 
    + \int_{g(\ell)}^{\ell} \left( \big( 1-f(\ell) \big) - \int_{\theta}^{\ell} \big( 1-f(y) \big) dy \right) e^{-\theta} d\theta 
    \geq \Gamma 
    ~,
\end{equation}
subject to the boundary conditions that $f(0) = 1 - \Gamma$ and $f(\ell_\max) = 1 - \frac{1}{e}$.

How to find such a function $g$?
We take a iterative approach.
Let us start with some $f^0$, say, $f^0(\ell) = \frac{1}{2}$ for all $\ell \ge 0$.
Then, we compute $g^1$ to be the best response to $f^0$, i.e., for any $0 \le \ell \le \ell_\max$:
\[
    \textstyle
    1-f^0(\ell) = \int_{g^1(\ell)}^\ell \big( 1-f^0(y) \big) dy 
    ~.
\]
For the aforementioned choice of $f^0$, we simply have $g^1(\ell) = \max \big\{ 0, \ell - 1 \big\}$.

\begin{figure}
    \begin{subfigure}[b]{.5\textwidth}
        \includegraphics[width=\textwidth]{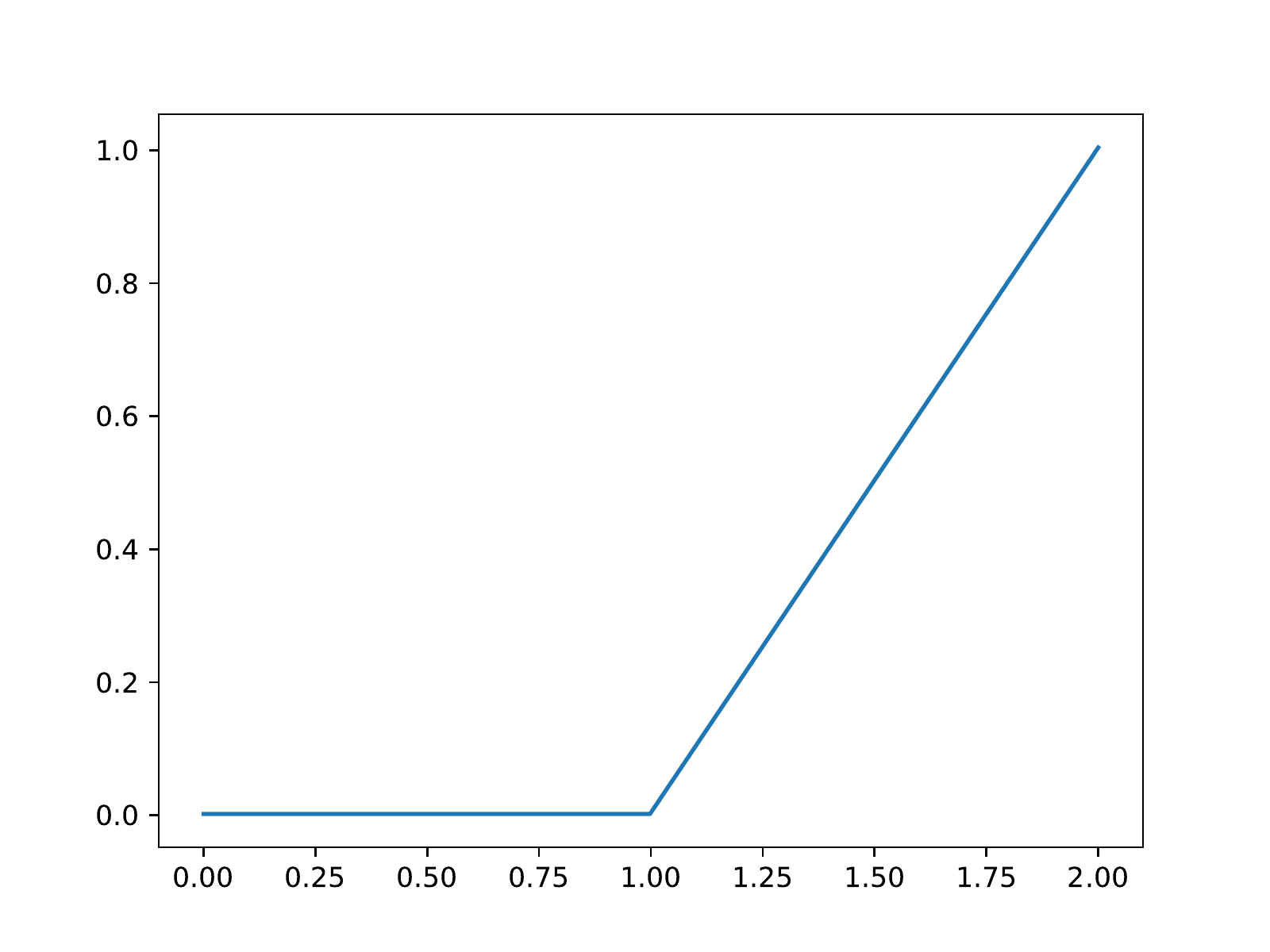}
        \caption{$g^1$}
    \end{subfigure}
    \begin{subfigure}[b]{.5\textwidth}
        \includegraphics[width=\textwidth]{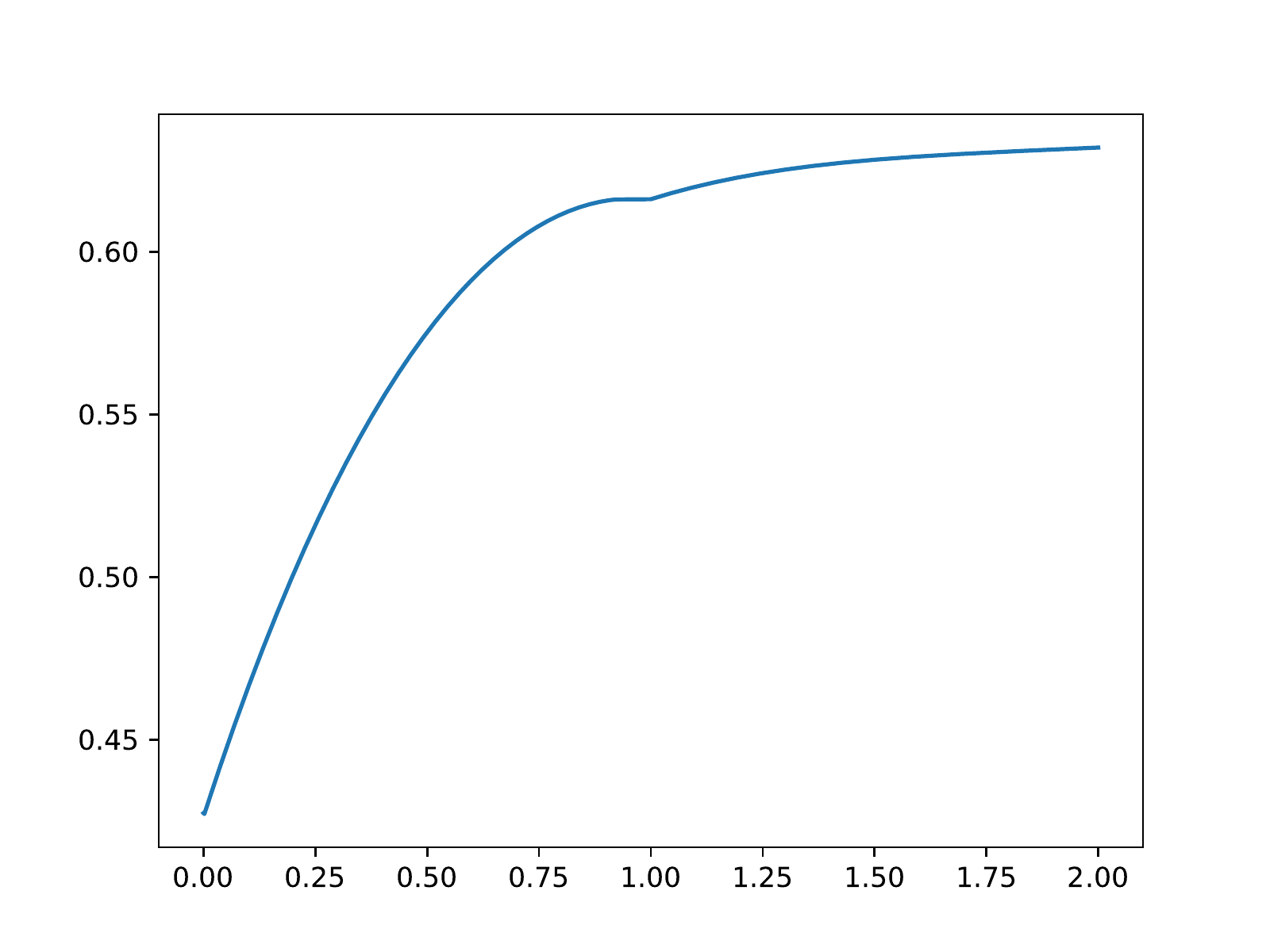}
        \caption{$f^1$}
    \end{subfigure}
    \begin{subfigure}[b]{.5\textwidth}
        \includegraphics[width=\textwidth]{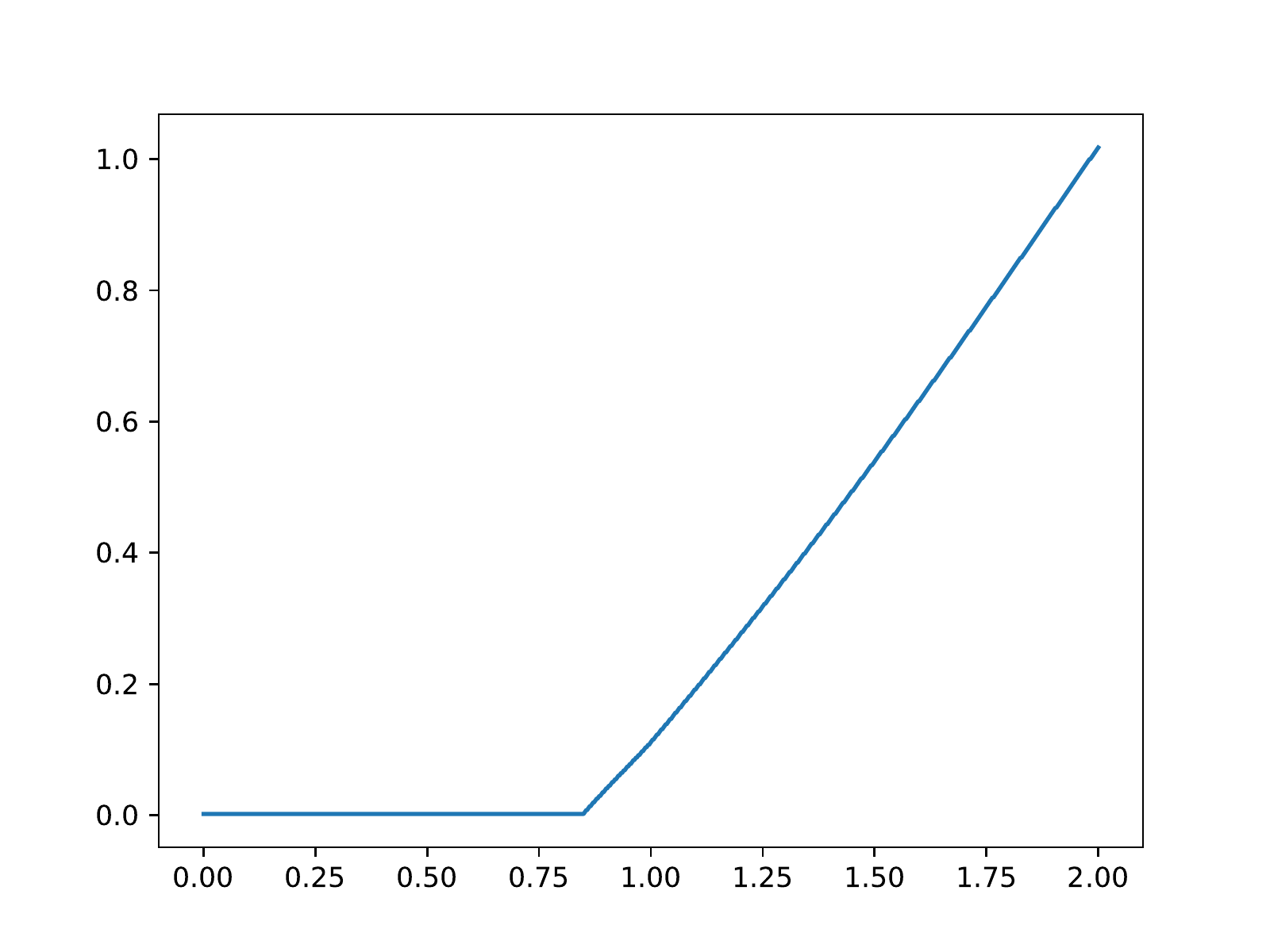}
        \caption{$g^2$}
    \end{subfigure}
    \begin{subfigure}[b]{.5\textwidth}
        \includegraphics[width=\textwidth]{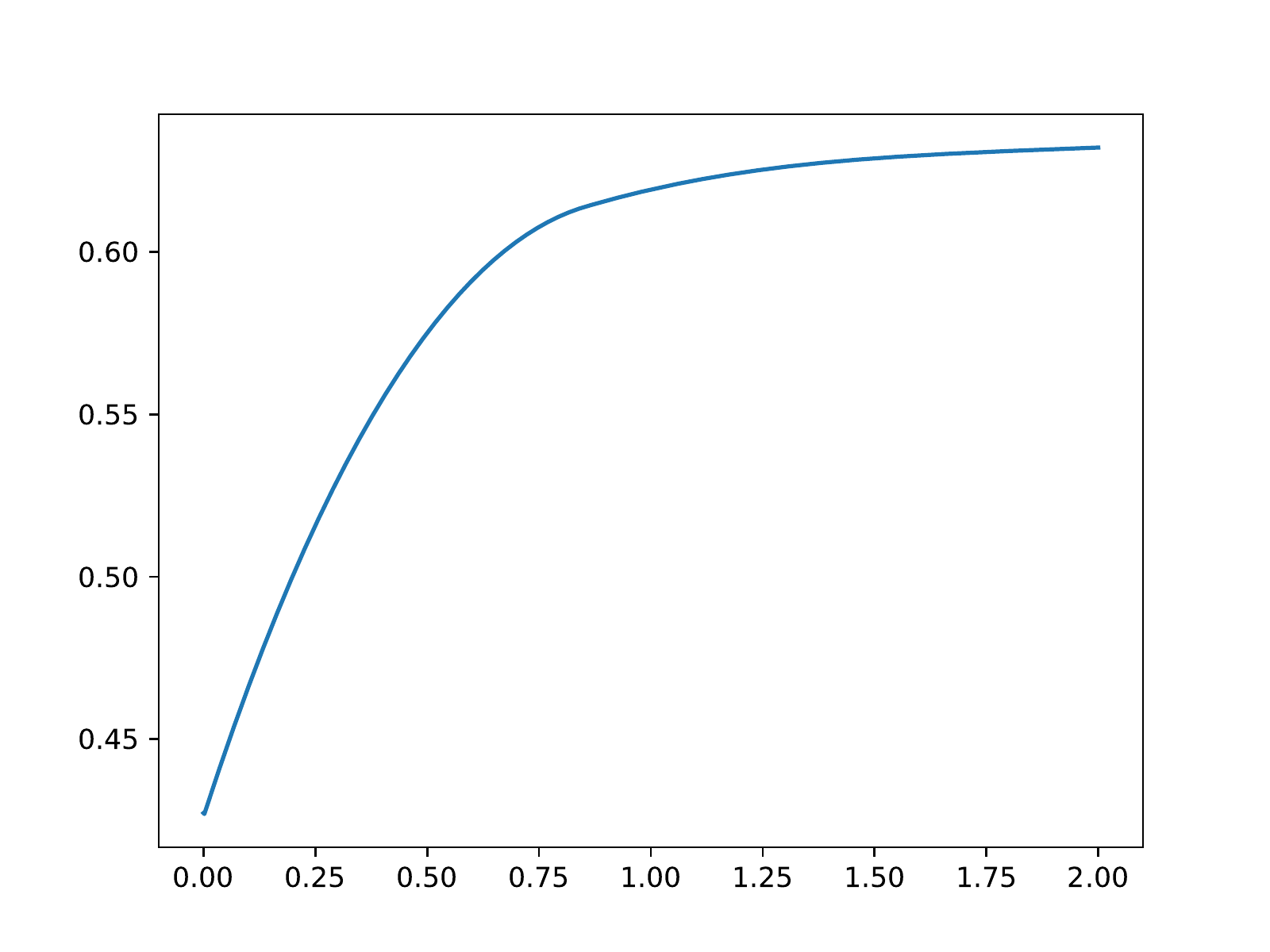}
        \caption{$f^2$}
    \end{subfigure}
    \begin{subfigure}[b]{.5\textwidth}
        \includegraphics[width=\textwidth]{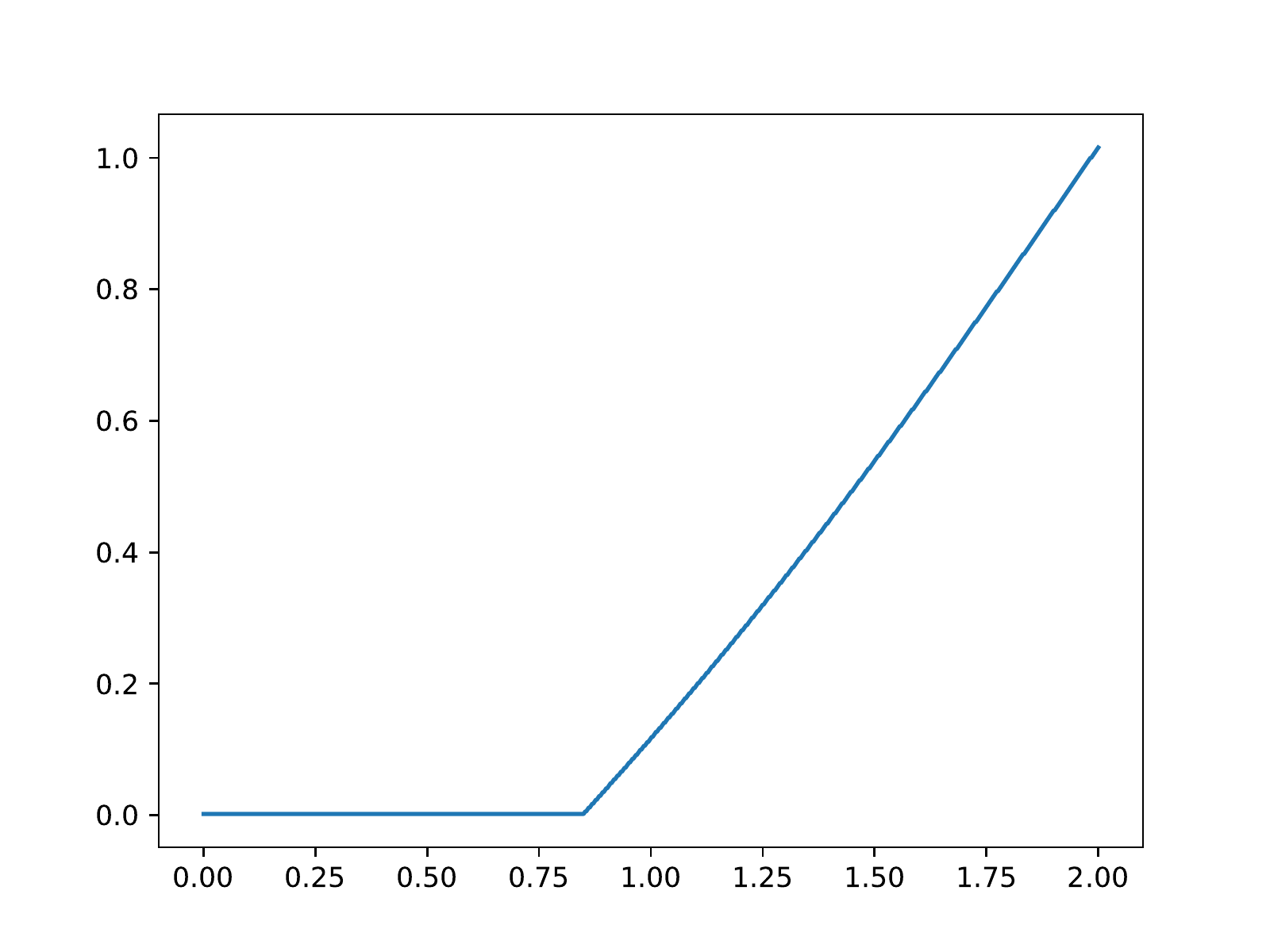}
        \caption{$g^3$}
    \end{subfigure}
    \begin{subfigure}[b]{.5\textwidth}
        \includegraphics[width=\textwidth]{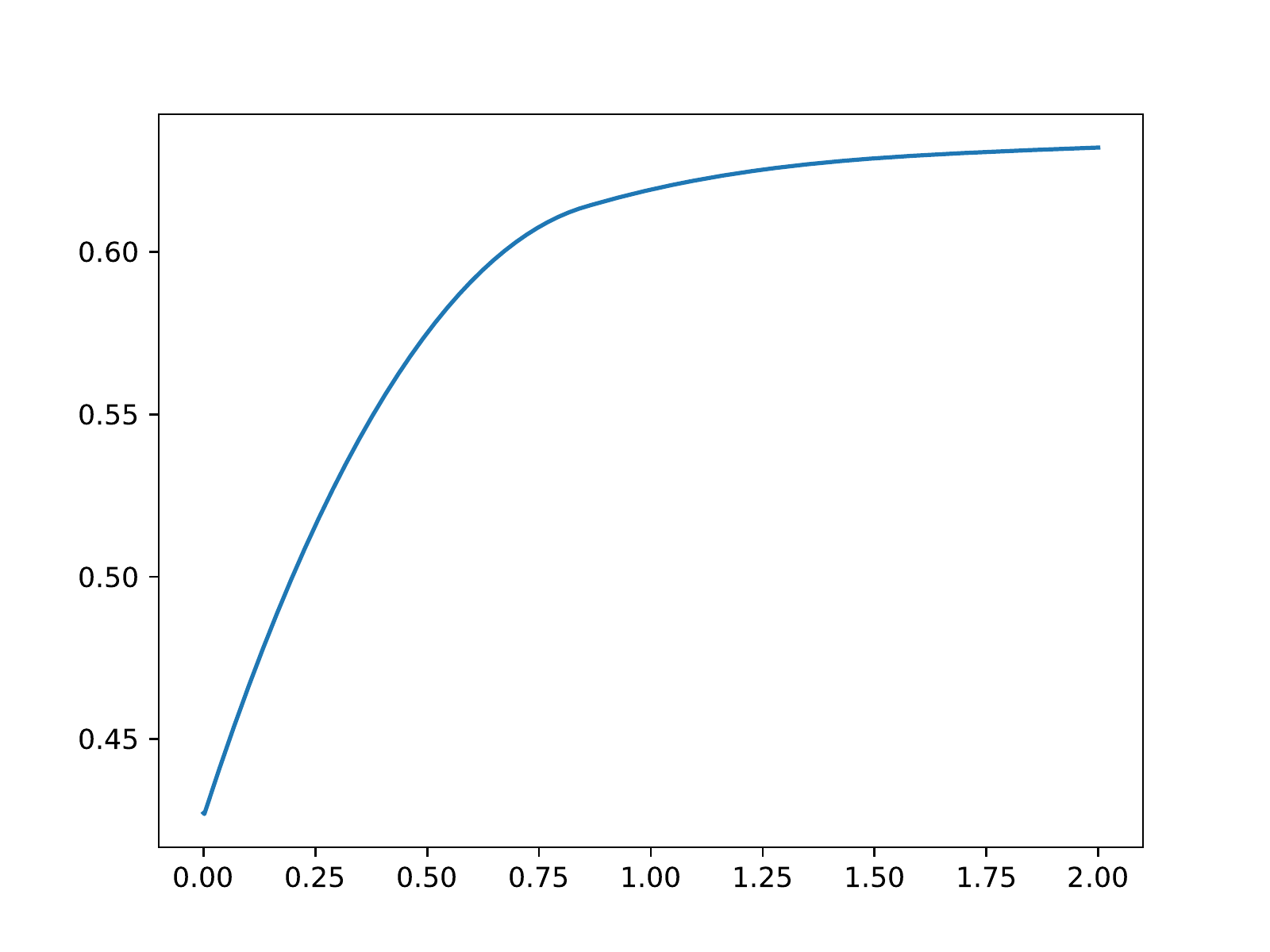}
        \caption{$f^3$}
    \end{subfigure}
    \caption{Results of the relaxed differential inequalities in the first three iterations of solving, with competitive ratios $\Gamma^1 \approx 0.5725971$, $\Gamma^2 \approx 0.5727709$, and $\Gamma^3 \approx 0.5727711$. Further optimizations do not improve the ratio in the first six digits.}
    \label{fig:iterative-de-solutions}
\end{figure}

Next, we find the optimal $f^1$ and the corresponding $\Gamma^1$a for the differential inequalities in Eqn.~\eqref{eqn:de-unequal-relaxed} w.r.t.\ $g^1$.

Further, we repeat this process and compute $g^2$ to be the best response to $f^1$, and then optimize $f^2$ and $\Gamma^2$ w.r.t.\ $g^2$, and so forth.
Throughout the process, we discretize both $f$ and $g$ appropriately so that the resulting competitive ratios $\Gamma^1$, $\Gamma^2$, etc.\ are truly lower bounds of the optimal competitive ratio.
This process converges after a few iterations.
See Figure~\ref{fig:iterative-de-solutions} for details.

\subsection{Alternating Path Argument: Failure Mode}
\label{sec:unequal-probabilities-alternating-path}

\begin{figure}
    \centering
    \includegraphics[width=.5\textwidth]{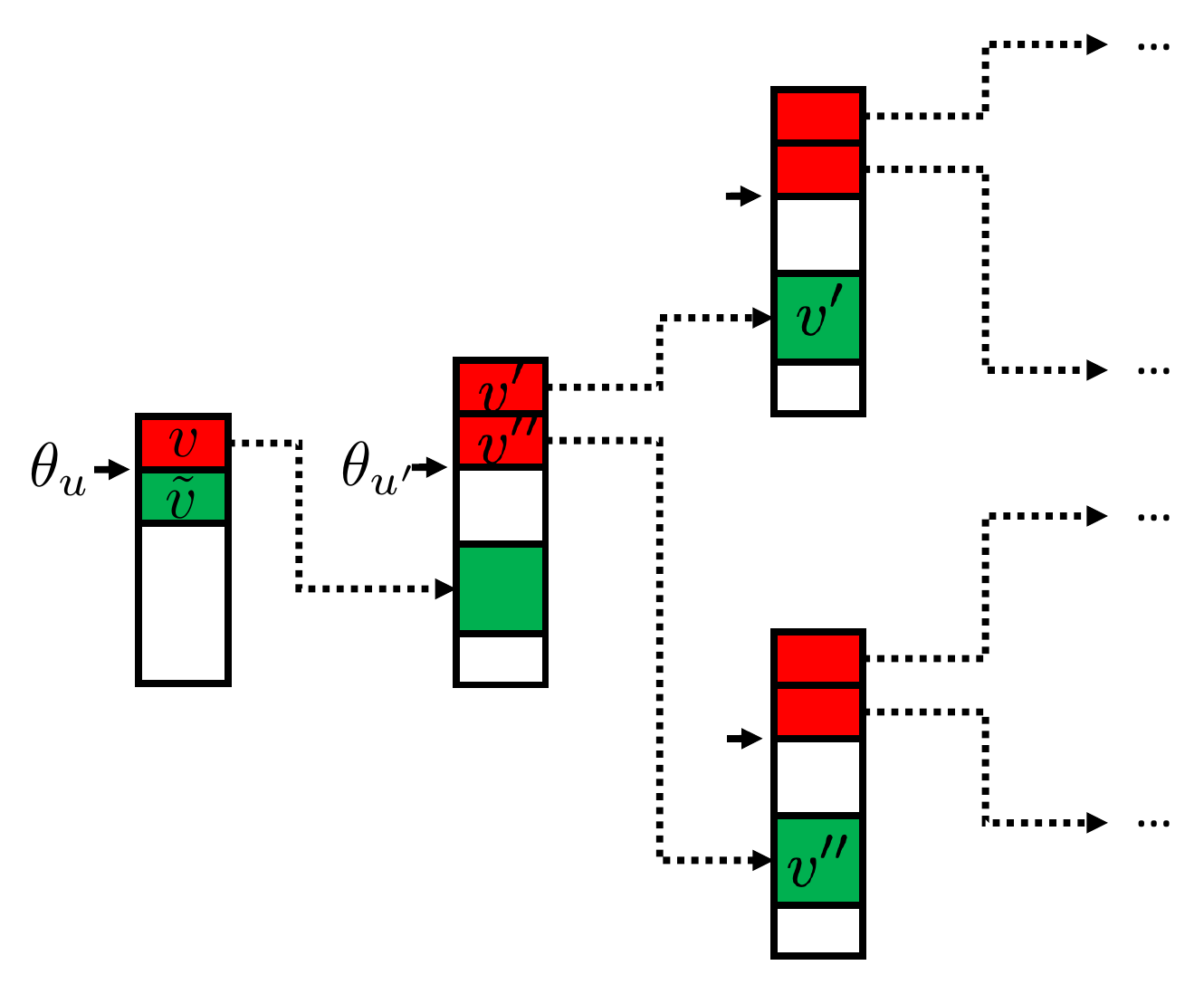}
    \caption{An illustrative picture of how the change of the match of a single online vertex can trigger a sequence of cascading subsequent changes.}
    \label{fig:alternating-path-failure}
\end{figure}

We start by recalling the bad instance for the alternating path argument.
Fix any offline vertex $u$ and the randomness related to the offline vertices other than $u$, i.e., their thresholds. 
First, consider the set of online vertices that are matched to $u$ if it never succeeds;
let $v$ and $\tilde{v}$ be the last and second last online vertices matched to $u$.
Then, consider what if $\tilde{v}$ succeeded and as a result $v$ could not be matched to $u$.
For concreteness, suppose $p_{uv} = \epsilon$ for some sufficiently small $\epsilon$.
It triggers a sequence of events:
\begin{enumerate}
    \item 
    First, $v$ may switch to another offline vertex, say $u'$, with whom its success probability is twice as large as that with $u$, i.e., $p_{u'v} = 2 \epsilon$.
    \item 
    Having $v$ matched to $u'$ increases the load of $u'$ by $2 \epsilon$.
    Further suppose that vertex $u'$ succeeds even before the change.
    Then, the last $2\epsilon$ amount of load that is matched to $u'$ before the change is must now be matched elsewhere.
    For concreteness, suppose this involves two vertices, say, $v'$ and $v''$, whose edges with $u'$ have success probabilities equal $\epsilon$.
    \item
    Each of $v'$ and $v''$ may not switch to another neighbor with whom they have an edge with success probabilities $2 \epsilon$.
    As a result, each of them forces a total load of $2\epsilon$, originally matched to the offline vertices they switch to, to be matched elsewhere. 
    \item This process may keep triggering more and more vertices to switch.
\end{enumerate}

Clearly, the changes triggered by the initial event that single online vertex, $v$, changes to another offline neighbor is no longer an alternating path.
Instead, the number of vertices as well as the sum of the success probabilities of their original matches grow exponentially as we go deeper into the sequence of events triggered by it.
Hence, the case of unequal probabilities cannot rely the invariants based on the alternating path argument.
See Figure~\ref{fig:alternating-path-failure} for an illustrative picture.

\section{Standard Matching LP}
\label{sec:std-lp}

This section explains in more details why a randomized online primal dual analysis using the standard matching LP fails. 
Let us consider the simplest version of the problem, i.e., the Stochastic Balance algorithm on unweighted instances with vanishing and equal probabilities.
We argue below that the randomized online primal dual, instantiated with the standard matching LP presented in Section~\ref{sec:model}, fails to give any competitive ratio that is better than the trivial one of $0.5$. 

Recall the standard matching LP and its dual:
\begin{align*}
    \textbf{StdLP:} \qquad \textrm{maximize} \quad & 
    \textstyle \sum_{(u,v) \in E} p_{uv} \cdot x_{uv} \\
	\textrm{subject to} \quad & 
    \textstyle \sum_{v : (u,v) \in E} p_{uv} \cdot x_{uv} \le 1 && \forall u \in U \\
	& 
    \textstyle \sum_{u : (u,v) \in E} x_{uv} \le 1 && \forall v \in V \\
    & x_{uv} \ge 0 && \forall (u, v) \in E
\end{align*}

The corresponding dual LP is as follows:
\begin{align*}
    \textbf{StdDual:} \qquad \textrm{minimize} \quad & 
    \textstyle \sum_{u \in U} \alpha_u + \sum_{v \in V} \beta_v \\
	\textrm{subject to} \quad & 
    \textstyle p_{uv} \cdot \alpha_u +  \beta_v \ge p_{uv} & & \forall (u,v) \in E \\
	& \alpha_u, \beta_v \ge 0 & & \forall u \in U, \forall v \in V
\end{align*}

Next, we explain the counterpart of Lemma~\ref{lem:randomized-primal-dual}, i.e., sufficient conditions that prove a competitive ratio of $\Gamma$ using the randomized online primal dual framework of \citet{DevanurJK/SODA/2013} introduced in Section~\ref{sec:prelimaries}. 
The proof is almost verbatim and therefore omitted.

\begin{lemma}
\label{lem:randomized-primal-dual-std-lp}
    An algorithm is $\Gamma$-competitive if:
    \begin{enumerate}
        \item \emph{Equal primal and dual objectives:} 
        \[
            \sum_{(u,v) \in E}  p_{uv} x_{uv} = \sum_{u \in U} \alpha_u + \sum_{v \in V} \beta_v 
            ~;
        \]
        \item \emph{Approximate dual feasibility:} For any $(u,v) \in E$, 
        \[
            \E \big[ p_{uv} \alpha_u +  \beta_v \big] \geq \Gamma \cdot p_{uv}
            ~.
        \]
    \end{enumerate}
\end{lemma}

In the equal probabilities case, $p_{uv}=p$ for all edges $(u,v)$. Recall the Stochastic Balance algorithm matches the online vertex $v$ to the unsuccessful neighbor $u$ with the least current load $\ell_u$. 
The dual assignment scheme via the standard LP is exactly the same as that in Section~\ref{sec:equal-probabilities-dual-unweighed}, which we restate below ($f : \R^+ \mapsto \R^+$ is a non-decreasing function to be determined by the analysis):
\begin{enumerate}
    \item Initially, let $\alpha_u = 0$ for all $u \in U$.
    \item On the arrival of an online vertex $v \in V$ and, thus, the corresponding dual variable $\beta_v$:
    \begin{enumerate}
        \item If $v$ is matched to $u \in U$, increase $\alpha_u$ by $p f(\ell_u)$ and let $\beta_v = p \big( 1 - f(\ell_u) \big)$. 
        \item If $v$ has no unsuccessful neighbor, let $\beta_v = 0$.
    \end{enumerate}
\end{enumerate}

When $p$ tends to zero, we have:
\[
    \alpha_u = F(\ell_u) = \int_0^{\ell_u} f(z) dz
    ~.
\]

Further, since the gain sharing procedure is splitting the gain of $p$ from each matched edge to the dual variables of its two endpoints, the first condition in Lemma~\ref{lem:randomized-primal-dual-std-lp} holds by definition. 
Hence, it suffices to show the second condition holds with the desired competitive ratio $\Gamma$, i.e., for any edge $(u,v) \in E$, over the randomness of the thresholds,
\[
    \E \big[ p \alpha_u + \beta_v \big] \geq \Gamma \cdot p
    ~.
\]

Next, fix the thresholds of all offline vertices except $u$, i.e., $\vec{\theta}_{-u}$ and consider the randomness of $\theta_u$. 
Recall that $\ell_u^{\infty}$ denotes the total load of $u$ when $\theta_u = \infty$, i.e., when all matchings to $u$ fail. 
We have the following characterization of the matching related to the endpoints $u$ and $v$ of any given edge $(u, v) \in E$ for different realization of $\theta_u$.
\begin{itemize}
    \item \textbf{Case 1: $\theta_u \ge \ell_u^\infty$.~}
        In this case, the matching is the same as the $\theta_u = \infty$ case since $u$ has enough budget to accommodate all the load matched to it.  
    \begin{itemize}
        \item The load of $u$ equals $\ell_u^\infty$.
        \item All the neighbors of $u$ including $v$ are matched to offline neighbors with load at most $\ell_u^\infty$ at the time of the matches, since $u$ is available with load at most $\ell_u^\infty$.
    \end{itemize}
    \item \textbf{Case 2: $0 \le \theta_u < \ell_u^\infty$.~}
        In this case, $u$ can only accommodate the load up to its threshold. 
        Some online vertices that were matched to $u$ when $\theta_u = \infty$ need to be matched elsewhere.
    \begin{itemize}
        \item The load of $u$ equals $\theta_u$.
        %
        \item The online vertex $v$ may become unmatched due to the alternating path argument in Lemma~\ref{lem:equal-probabilities-structural}. 
    \end{itemize}
\end{itemize}

\paragraph{Failure of the Standard Matching LP.}
The characterization of the online side in the second case above is the main bottleneck of the analysis using the standard matching LP. 
Let us compare it with the one using the configuration LP in Section~\ref{sec:equal-probabilities}. 
The contribution from the online side above drops to zero as soon as $u$'s threshold $\theta_u$ falls below $\ell_u^\infty$ in the above characterization, while the counterpart in Section~\ref{sec:equal-probabilities} decreases gracefully and does not drop to zero until $\theta_u$ is much lower than $\ell_u^\infty$.
What happens in the analysis using the configuration LP can be viewed as an amortization among the subset of online vertices under consideration.

\bigskip

We next proceed with the argument using the standard matching LP and show that the above bottleneck leads to a competitive ratio of only $0.5$.

\paragraph{Contribution from $\alpha_u$.}
The lower bound of the contribution from the offline side is a verbatim copy of the counterpart in Section~\ref{sec:equal-probabilities}.
We reiterate the lemma below and omit the proof.

\begin{lemma}
    \label{lem:equal-probabilities-alpha-std-lp}
    For any thresholds $\vec{\theta}_{-u}$ of offline vertices other than $u$ and the corresponding $\ell_u^\infty$:
    \[
        \textstyle
        \E _{\theta_u} \big[ \alpha_u | \vec{\theta}_{-u} \big] \geq \int_0^{\ell_u^\infty} e^{-\theta_u} f(\theta_u) d\theta_u 
        ~.
    \]
\end{lemma}

\paragraph{Contribution from $\beta_v$.}
For the online side, we have the following lower bound. 
Interested readers may compare it with the one in Section~\ref{sec:equal-probabilities} to see how much weaker the bound becomes.
%

\begin{lemma}
    \label{lem:equal-probabilities-beta-std-lp}
    For any thresholds $\vec{\theta}_{-u}$ of offline vertices other than $u$ and the corresponding $\ell_u^\infty$:
    \[
        \textstyle
        \E _{\theta_u} \big[ \beta_v | \vec{\theta}_{-u} \big]
        \geq e^{-\ell_u^{\infty}} p \big( 1-f(\ell_u^{\infty}) \big)
        ~.
    \]
\end{lemma}

\begin{proof}
    If $\theta_u \ge \ell_u^{\infty}$, which happens with probability $e^{-\ell_u^\infty}$, all the neighbors of $u$ including $v$ are matched to offline vertices with load at most $\ell_u^{\infty}$ at the time of the matches, which means that $\beta_v$ is at least:
    \[
        p \big( 1-f(\ell_u^{\infty}) \big)
        ~.
    \]
    
    If $ \theta_u < \ell_u^{\infty}$, $v$ might be left unmatched by Lemma~\ref{lem:equal-probabilities-structural}, thus leading to $\beta_v =0$.
    
    Putting together proves the lemma.
\end{proof}

Combining Lemma~\ref{lem:equal-probabilities-alpha-std-lp} and Lemma~\ref{lem:equal-probabilities-beta-std-lp}, we get,
\[
        \textstyle
        \E \big[p \alpha_u + \beta_v \big] \ge
        p \left( \int_0^{\ell_u^\infty} e^{-\theta_u} f(\theta_u) d\theta_u +  e^{-\ell_u^{\infty}} \big( 1-f(\ell_u^\infty) \big) \right)
        ~.
    \]

Hence, to show a competitive ratio $\Gamma$, it suffices to find a non-decreasing function $f$ that satisfies the following differential inequality.
For ease of notations, we write $\ell_u^\infty$ as $\ell$ and $\theta_u$ as $\theta$.
For any $\ell \ge 0$, we need:
\begin{equation}
    \label{eqn:de-unequal-std-lp}
    \textstyle
    \int_0^{\ell} e^{-\theta} f(\theta) d\theta + e^{-\ell} \big( 1-f(\ell) \big) 
    \geq \Gamma
    ~.
\end{equation}

Solving the above differential inequality to maximize $\Gamma$ gives $f(x) = 0.5$ for any $x \geq 0$ and $\Gamma = 0.5$. 
Thus, we only get a competitive ratio of $\Gamma = 0.5$.

\bibliographystyle{plainnat}
\bibliography{matching}

\end{document}